\ifdefined\ONECOLUMN
	\documentclass[journal, 12pt, onecolumn, draftclsnofoot, a4paper]{IEEEtran}
\else
	\documentclass[journal, a4paper, twoside]{IEEEtran}
\fi

\usepackage{fancyhdr}
\usepackage{cite}
\usepackage{amsmath,color}
\usepackage{amssymb}
\usepackage{latexsym}
\usepackage{longtable}
\usepackage{mathrsfs}
\usepackage{color}
\usepackage{multicol}
\usepackage{amsfonts}
\usepackage{array}
\usepackage[nomain,acronym,automake,nopostdot]{glossaries}
\definecolor{sred}{RGB}{130,51,0}
\definecolor{spurple}{RGB}{155,0,155}
\definecolor{sblue}{RGB}{0,0,0}

\ifCLASSINFOpdf
   \usepackage[pdftex]{graphicx}
   \graphicspath{{../pdf/}{../jpeg/}}
   \DeclareGraphicsExtensions{.pdf,.jpeg,.png}
\else
   \usepackage[dvips]{graphicx}
   \graphicspath{{../eps/}}
   \DeclareGraphicsExtensions{.eps}
\fi

\newtheorem{thm}{Theorem}

\newtheorem{prop}{Proposition}
\newtheorem{defn}{Definition}

\newcommand{\matc}[1]{\mbox{\boldmath $\mathcal{#1}$}}

\newcommand{\figref}[1]{Fig. \ref{#1}}

\makeglossaries
\newacronym{ils}{ILS}{integer least-squares}
\newacronym{snr}{SNR}{signal-to-noise ratios}
\newacronym{wl}{WL}{widely linear}
\newacronym{mlsd}{MLSD}{maximum-likelihood sequence detection}
\newacronym{map}{MAP}{maximum a posteriori}
\newacronym{lmmse}{LMMSE}{linear minimum mean-square error}
\newacronym{ut}{UT}{user terminal}
\newacronym{bs}{BS}{base station}
\newacronym{ann}{ANN}{artificial neural networks}
\newacronym{dnn}{DNN}{deep neural network}
\newacronym{snn}{SNN}{super neural networks}
\newacronym{gmnn}{G-MNN}{giant modular neural network}
\newacronym{mnn}{MNN}{modular neural network}
\newacronym{pic}{PIC}{parallel interference cancellation}
\newacronym{mimo}{MIMO}{multiple-input multiple-output}
\newacronym{mud}{MUD}{multiuser detection}
\newacronym{amp}{AMP}{approximate message passing}
\newacronym{zf}{ZF}{zero forcing}
\newacronym{mf}{MF}{matched filter}
\newacronym{cdma}{CDMA}{code division multiple access}
\newacronym{noma}{NOMA}{non-orthogonal multiple access}
\newacronym{csir}{CSIR}{channel state information at the receiver}
\newacronym{sic}{SIC}{successive interference cancellation}
\newacronym{cdf}{CDF}{cumulative distribution function}
\newacronym{ncl}{NC-Learning}{non-coherent learning}
\newacronym{cel}{CE-Learning}{channel equalized learning}
\newacronym{dcl}{DC-Learning}{direct coherent learning}
\newacronym{bp}{BP}{back propagation}
\newacronym{vq}{VQ}{vector quantization}
\newacronym{awgn}{AWGN}{additive white Gaussian noise}
\newacronym{ai}{AI}{artificial intelligence}
\newacronym{dsp}{DSP}{digital signal processing}
\newacronym{csi}{CSI}{channel state information}
\newacronym{nn}{NN}{nearest-neighbor}
\newacronym{clnn}{CL-NN}{cluster-level nearest neighbor}
\newacronym{clknn}{CL-kNN}{cluster-level kNN}
\newacronym{ber}{BER}{bit error rate}
\newacronym{mai}{MAI}{multiple access interference}
\begin{document}
%
\title{{\color{sblue}A Modular Neural Network {\color{sblue}Based} Deep Learning Approach for MIMO Signal Detection}}
%
%
%

\author{Songyan Xue, Yi Ma, Na Yi, and Terence E. Dodgson
\thanks{\hrule\vspace{5pt} This work was supported in part by the European Union Horizon 2020 5G-DRIVE project (Grant No. 814956), in part by Airbus Defense and Space, and in part by the UK 5G Innovation Centre (5GIC).}
\thanks{Songyan Xue, Yi Ma and Na Yi are from the Institute for Communication Systems (ICS), University of Surrey, Guildford, England, GU2 7XH. E-mail: (songyan.xue, y.ma, n.yi)@surrey.ac.uk. Tel.: +44 1483 683609.} \thanks{Terence E. Dodgson is from the Airbus Defense and Space, Portsmouth, England, PO3 5PU. E-mail: terence.dodgson@airbus.com.}}%

\markboth{}%
{}

\maketitle

\begin{abstract}
In this paper, {\color{sblue}we reveal that} artificial neural network (ANN) assisted multiple-input multiple-output (MIMO) {\color{sblue}signal} detection {\color{sblue}can be} modeled as ANN-assisted lossy vector quantization (VQ), named MIMO-VQ, {\color{sblue}which is basically} a joint statistical channel quantization and signal quantization procedure. {\color{sblue}It is found that} the quantization loss increases linearly with the number of transmit antennas, and thus MIMO-VQ scales poorly with the size of MIMO. Motivated by this {\color{sblue}finding}, {\color{sblue}we {\color{sblue}propose} a novel modular neural network based approach, {\color{sblue}termed} MNNet, where the {\color{sblue}whole} network is formed by a set of pre-defined {\color{sblue}ANN} modules. The {\color{sblue}key of ANN} module design {\color{sblue}lies in} the {\color{sblue}integration} of parallel interference cancellation {\color{sblue}in the MNNet, which linearly reduces the interference (or equivalently the number of transmit-antennas) along the feed-forward propagation; and so as the quantization loss}. Our {\color{sblue}simulation} results show that the MNNet approach largely improves the {\color{sblue}deep-learning} capacity with near-optimal performance in various {\color{sblue}cases}}. {\color{sblue}Provided that MNNet} is well modularized, the learning procedure does not need to be applied on the entire network as a whole, but rather at the modular level. {\color{sblue}Due to this reason, MNNet has the advantage of much lower learning complexity than other deep-learning based MIMO detection approaches}.
\end{abstract}


\IEEEpeerreviewmaketitle

\section{Introduction}
\IEEEPARstart{C}{onsider} the discrete-time equivalent baseband signal model of the wireless \gls{mimo} channel with $M$ transmit antennas and $N$ receive antennas ($N\geq M$)

\begin{equation}\label{eq01}
\mathbf{y}=\mathbf{H}\mathbf{x}+\mathbf{v}
\end{equation}
where we define
\begin{itemize}
\item[$\mathbf{y}$:] the received signal vector with $\mathbf{y}\in \mathbb{C}^{N\times 1}$;
\item[$\mathbf{x}$:] the transmitted signal vector with $\mathbf{x}\in \mathbb{C}^{M\times 1}$. Each element of $\mathbf{x}$ is independently drawn from a finite-alphabet set consisting of $L$ elements,  with the equal probability, zero mean, and identical variance $\sigma_x^2$;
\item[$\mathbf{H}$:] the random channel matrix with $\mathbf{H}\in\mathbb{C}^{N\times M}$;
\item[$\mathbf{v}$:] the additive Gaussian noise with $\mathbf{v}\sim CN(0,\sigma_v^2\mathbf{I})$, and $\mathbf{I}$ is the identity matrix.
\end{itemize}
{\color{sblue} The fundamental aim is to find the closest lattice to $\mathbf{x}$ based upon $\mathbf{y}$ and $\mathbf{H}$, which is a classical problem in the scope of {\color{sblue}signal processing for communication. The optimum performance can be achieved through \gls{mlsd} algorithms. The downside of \gls{mlsd} algorithms lies in their exponential computation complexities. On the other hand, low-complexity linear detection algorithms such as \gls{mf}, \gls{zf} and \gls{lmmse}, are often too sub-optimum. This has motivated enormous research efforts towards the best performance-complexity trade-off of MIMO signal detection}; please see \cite{1266912} for an overview.

Recent advances towards the \gls{mimo} signal detection lie in the use of {\color{sblue}deep} learning. The basic idea is to train \gls{ann} as a black box so as to develop its ability of signal detection. The input of \gls{ann} {\color{sblue}is often a concatenation of received signal vector and \gls{csi}, i.e., $\mathbf{y}$ together with $\mathbf{H}$ or more precisely its vector-equivalent version $\breve{\mathbf{h}}$ \cite{DBLP:journals/corr/OSheaEC17}; as depicted in \figref{fig1}-(a). A relatively comprehensive state-of-the-art review can be found in \cite{8382166,8755300}. }
{\color{sblue}Notably, a detection network (DetNet) was proposed in \cite{8642915}, with their key idea to} unfold iterations of projected gradient descent algorithm into deep neural networks. In \cite{DBLP:journals/corr/abs-1812-10044}, a trainable projected gradient detector is proposed for {\color{sblue}overloaded massive-MIMO systems}.  Moreover, in \cite{8646357},  a model-driven deep learning approach, named OAMP-Net, is proposed by incorporating deep learning into the orthogonal approximate message-passing (OAMP) algorithm.}
\begin{figure}[t]
\begin{center}
\includegraphics[width=0.8\columnwidth]{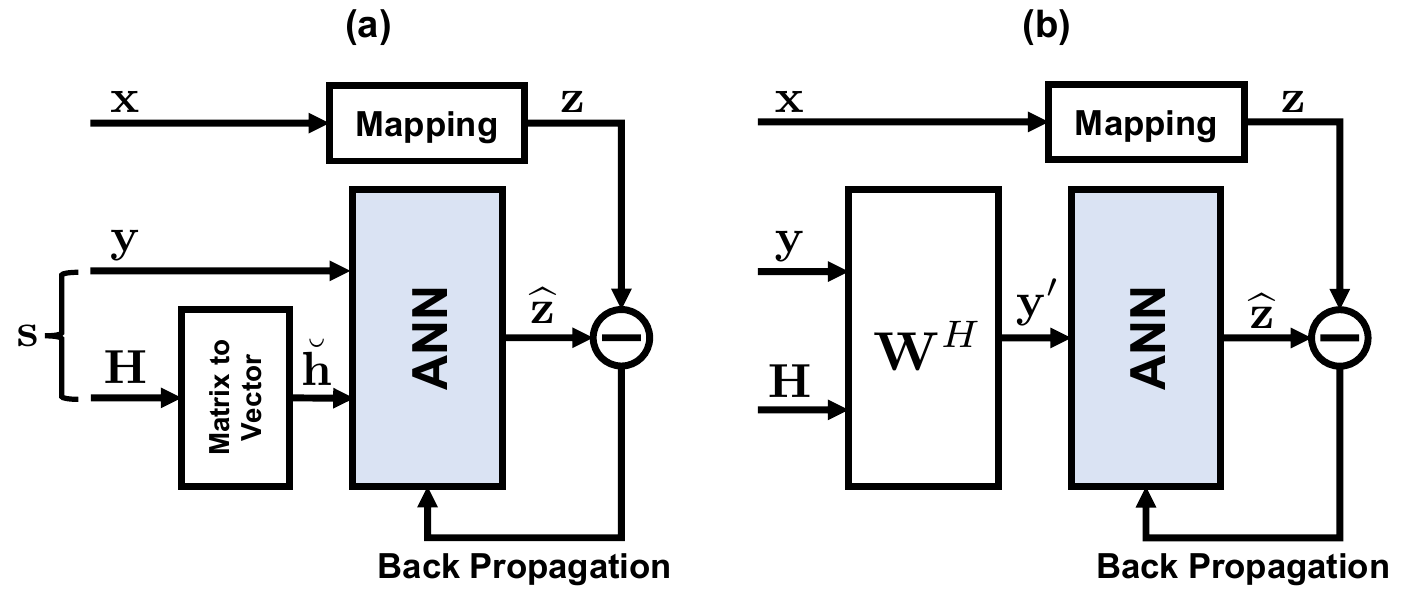}
\caption{{\gls{ann}-assisted coherent \gls{mimo} detection: (a) the \gls{ann} input is a combination of received signal and \gls{csi}; (b) the \gls{ann} input is a linearly filtered received signal.}}\label{fig1}
\end{center}
\end{figure}

Despite an increasingly hot topic, there is an ongoing debate on the use of machine learning for communication systems design, particularly on the modem side \cite{DBLP:journals/corr/abs-1901-08329}. A straightforward question would be: modem based on the simple linear model \eqref{eq01} can be well optimized by hand-engineered means; what are additional values machine learning could bring to us; and what will be the cost? Basically, it is not our aim to join the debate with this paper. However, we find it useful to highlight a number of the key features of the \gls{ann}-assisted \gls{mimo} detection based on published results, as they well motivated our work. Specifically, \gls{ann}-assisted \gls{mimo} {\color{sblue} detection} has the following remarkable advantages:

\subsubsection{Parallel computing ready} 
\gls{ann}-assisted \gls{mimo} receivers are mostly based upon feed-forward neural networks, which have a parallel computing architecture in nature \cite{Takefuji:1992:NNP:128941}. It fits into the trend of high-performance computing technologies that highly rely on parallel processing power to improve the computing speed, the capacity of multi-task execution as well as the computing energy-efficiency. This is an important feature as it equips the receiver with a great potential of providing ultra-low latency and energy-efficient signal processing that is one of key requirements for future wireless networks \cite{DBLP:journals/corr/abs-1812-02858,Yazar2019}.

\subsubsection{Low receiver complexity}
\gls{ann}-assisted \gls{mimo} receivers only involve a number of matrix multiplications, depending on {\color{sblue} the number of} hidden layers involved in the feed-forward procedure. Moreover, they bypass channel matrix inversions or factorizations (such as singular-value decomposition or QR decomposition) which are needed for most of conventional \gls{mimo} receivers including linear \gls{zf}, \gls{lmmse} \cite{4686826}, \gls{sic} \cite{5910122}, sphere decoding \cite{1341066}, and many others \cite{6375940}. The complexity for channel matrix inversions or factorizations is around $\mathcal{O}(NM^2)$ with $(N\geq M)$. This can be affordable complexity for current real-time \gls{dsp} technology as long as the size of the channel matrix is reasonably small (e.g. $M=10$ or smaller) \cite{8416771}. However, such rather prevents the development of future \gls{mimo} technologies that aim to exploit increasingly the spatial degree of freedom for spectral efficiency. One might argue for the use of high-performance parallel computing technologies to mitigate this bottleneck. However, there is a lack of parallel computing algorithms for matrix inversions or something equivalent to this date.

\subsubsection{Good performance-complexity trade-off}
it has been demonstrated that extensively trained \gls{ann}-assisted receivers can be well optimized for their training environments (or channel models). For instance in \cite{8054694}, the performance of \gls{ann}-assisted receiver is very close to that of {\color{sblue}the \gls{mlsd}}. This is rather encouraging result as it reaches a good trade-off between the receiver complexity and the performance.

Despite their advantages, current \gls{ann}-assisted \gls{mimo} receivers face a number of fundamental challenges.
\setcounter{subsubsection}{0}
\subsubsection{\gls{ann} learning scalability}  in \gls{mimo} fading channels, the \gls{ann} blackbox in \figref{fig1}-(a) has its learning capacity rapidly degraded with the growth of transmit antennas \cite{8761999}. Current approaches to mitigate this problem are by means of training the \gls{ann} with channel equalized signals; as depicted in \figref{fig1}-(b). However, in this case,  \gls{ann}-assisted receivers are not able to exploit maximally the spatial diversity-gain due to the multiuser orthogonalization enabled by channel equalizers \cite{7244171}, and as a consequence, their performances go far away from the \gls{mlsd} receiver.

\subsubsection{Learning expenses}
the \gls{ann} learning process often involves very expensive computation cost and energy consumption as far as the conventional computing architecture is concerned. The aim of reducing \gls{ann} learning expenses has recently motivated a new research area on {\em non von Neumann} computing architectures \cite{10.1007/11802839_32}.

\subsubsection{Training set over-fitting}
an \gls{ann}-assisted receiver trained for a specific wireless environment (or channel model) is often not suitable for another environment (or channel model) \cite{Hawkins2004ThePO}. This issue can be viewed more positively. For instance in urban areas, access points (AP), such as LTE eNBs or 5G gNBs, often have their physical functions optimized for local environments \cite{1181869}. It could be an advantage if \gls{mimo} transceivers integrated in APs can also be optimized for their local environments through machine learning.

\subsubsection{The black-box problem}
it is well recognized that the black-box model is challenging the system reliability and maintenance job. ``{\em How to make \gls{ai} more describable or explainable?}'' becomes an increasingly important research topic in the general \gls{ai} domain \cite{article1,article2,8466590}. 

Again, we stress that the objective of this paper is not to offer a comparison between machine learning and hand-engineered approaches in the scope of \gls{mimo} detection. {\color{sblue}Instead,  our aim is to develop a deeper understanding of the fundamental behaviors of \gls{ann}-assisted \gls{mimo} detection, based on which we can find a scalable approach for that.
 {\color{sblue}Major} contributions of this paper {\color{sblue}include}:
\begin{itemize}
\item An {\color{sblue}extensive} study of the \gls{ann}-assisted \gls{mimo} detection model in terms of its performance and scalability. By borrowing the basic concept from the \gls{ann} \gls{vq} model used for source encoding in \cite{Ahalt93}, our work reveals that \gls{ann}-assisted \gls{mimo} detection can be modeled as the \gls{ann}-assisted lossy \gls{vq}, named \gls{mimo}-\gls{vq}, which is naturally a joint statistical channel quantization and massage quantization procedure. By means of the \gls{nn} mapping, \gls{mimo}-\gls{vq} is shown to be integer least-squares optimum with its complexity {\color{sblue}growing} exponentially with the number of transmit antennas;

\item The analysis of codebook-vector combination approaches including \gls{clnn} and \gls{clknn}. It is shown that both approaches {\color{sblue}require much less neurons than the \gls{nn}-based \gls{mimo}-\gls{vq}}. However, those low-complexity approaches introduce considerable loss to the statistical channel quantization (equivalently increase of channel ambiguity for the message quantization), which results in learning inefficiency particularly at higher \gls{snr};

\item The development of MNNet, a novel \gls{mnn} based \gls{mimo} detection approach that {\color{sblue}can achieve} near-optimal performance with much lower computational complexity {\color{sblue}than state-of-the-art} deep-learning based \gls{mimo} detection methods. The {\color{sblue}idea} of MNNet {\color{sblue}lies in the integration of \gls{pic} {in deep \gls{mnn}, which linearly reduces the interference (or equivalently the number of transmit-antennas) along the feed-forward propagation; and so as the quantization loss. Moreover, the learning procedure of MNNet can be applied at the modular level; and this  largely reduces the learning complexity}}.
\end{itemize}}

The rest of this paper is organized as follows. Section II presents the {\color{sblue}novel} \gls{ann}-assisted \gls{mimo}-\gls{vq} model. Section III presents fundamental behaviors of the \gls{ann}-assisted \gls{mimo}-\gls{vq} as well as the learning scalability. Section IV presents the novel MNNet {\color{sblue}approach} as well as {\color{sblue}the} performance and complexity analysis. Simulation results are presented in Section V; and finally, Section VI draws the conclusion.

\section{{\color{sblue}\gls{mimo}-\gls{vq} Model for The \gls{ann}-Assisted \gls{mimo} Detection}}
\label{sec2}
\subsection{Concept of Vector Quantization}\label{sec2a}
A vector quantizer is a statistical source encoder, which aims to efficiently and compactly represent the original signal using a digital signal. The encoded signal shall retain the essential information contained in the original signal \cite{Gersho}. More rigorously, \gls{vq} can be described by
\begin{defn}[\cite{Ahalt93}]\label{def1}
Given an arbitrary input vector $\mathbf{s}\in\mathbb{R}^{K\times 1}$ and a finite set $\matc{A}=\{0, 1, ...,  J-1\}$, \gls{vq} forms a mapping: $\mathbf{s}\longrightarrow j$, with $j\in\matc{A}$ denoting the index of codebook vectors. Given $p_j$ the probability when $\mathbf{s}$ is mapped to $j$,  the average rate of the codebook vectors is $R=-\sum_{j=0}^{J-1}p_j\log_2p_j$. Since each input vector has $K$ components, the number of bits required to encode each input vector component is $(R)/(K)$. Moreover, the compression rate of \gls{vq} is: $r=(R)/\hbar(\mathbf{s})$, where $\hbar(\mathbf{s})$ is the entropy of $\mathbf{s}$. 
\end{defn}

The aim of \gls{vq} is to find the optimum codebook which minimizes the average quantization loss (distortion) given the codebook size $J$. 
\begin{defn}\label{def2}
Given $\hat{\mathbf{s}}_j\in\mathbb{R}^{K\times1}$ the reconstructed input vector (or called the anchor vector) corresponding to $j$, and $\Delta(\mathbf{s}, \hat{\mathbf{s}}_j)$ the quantization loss when $\mathbf{s}$ is mapped to $j$, the average quantization loss is 
\begin{equation}\label{eq02}
\overline{\Delta}=\underset{\mathbf{s}}{\mathbb{E}}\left(\sum_{j=0}^{J-1}p_j\min_j\Delta(\mathbf{s}, \hat{\mathbf{s}}_j)\right)
\end{equation}
where $\underset{\mathbf{s}}{\mathbb{E}}(\cdot)$ denotes the expectation over $\mathbf{s}$. 
\end{defn}

By means of minimizing the average quantization loss \eqref{eq02},  \gls{vq} effectively partitions input vectors into $J$ clusters, and forms anchor vectors $\hat{\mathbf{s}}_j, \forall j$; such is called the Voronoi partition. 

\subsection{\gls{ann}-Assisted Vector Quantization}\label{sec2b}
\gls{ann} architecture for \gls{vq} is rather straightforward. Consider a neural network having $J$ neurons, with each yielding a binary output $\hat{z}_j\in\{0, 1\}$. Let $\mathbf{w}_j$ be the weighting vector for the $j^\text{th}$ neuron. \gls{ann} can measure the quantization loss $\Delta(\mathbf{s}, \mathbf{w}_j)$ and apply the \gls{nn} rule to determine the output $\hat{z}_j$, i.e.,
\begin{equation}\label{eq03}
\hat{z}_j=\left\{\begin{array}{cl}
1,&\Delta(\mathbf{s}, \mathbf{w}_j)=\min\Delta(\mathbf{s}, \mathbf{w}_l), l=0, ..., J-1\\
0,&\text{otherwise}
\end{array}\right.
\end{equation}
Right after the \gls{ann} training, we let $\hat{\mathbf{s}}_j=\mathbf{w}_j$. Such shows how input vectors are optimally partitioned into $J$ clusters.

The \gls{ann} training procedure often starts with a random initialization for the weighting vectors $\mathbf{w}_j, \forall j$. Then, the training algorithm iteratively adjusts the weighting vectors with
\begin{equation}\label{eq04}
\mathbf{w}_j(i+1)=\mathbf{w}_j(i)+\beta(i)(\mathbf{s}-\mathbf{w}_j(i))z_j(i)
\end{equation}
where $i$ denotes the number of iterations, and $\beta(i)$ is the adaptive learning rate, which typically decreases with the growth of iterations. This is the classical Hebbian learning rule for the competitive unsupervised learning \cite{hebb-organization-of-behavior-1949}, which has led to a number of enhanced or extended versions such as the Kohonen self-organizing feature map algorithm \cite{Feng88}, learning VQ \cite{Kohonen90}, and frequency selective competitive learning algorithm \cite{Grossberg76}. Nevertheless, all of those unsupervised learning algorithms go beyond the scope of our research problem, and thus we skip the detailed discussion in this paper. 

\subsection{\gls{ann}-Assisted \gls{mimo} Vector Quantization}\label{sec2c}
Let us form an input vector $\mathbf{s}=[\breve{\mathbf{h}}^T,\mathbf{y}^T]^T$, {\color{sblue}where $\mathbf{y}$ is the received signal vector; $\breve{\mathbf{h}}$ is the \gls{csi} vector as already defined in Section I; and $[\cdot]^T$ stands for the matrix or vector transpose. All notations are now redefined in the real domain, i.e., $\mathbf{y}\in\mathbb{R}^{N\times 1}$, $\mathbf{H}\in\mathbb{R}^{N\times M}$, $\mathbf{x}\in\mathbb{R}^{M\times 1}$.} {\color{sblue}Since $\breve{\mathbf{h}}$ is the vector counterpart of $\mathbf{H}$, the length of $\breve{\mathbf{h}}$ is $(NM)$, and consequently the length of $\mathbf{s}$ is: $K=N(M+1)$.
Note that the conversion from complex signals to their real-signal equivalent version doubles the size of corresponding vectors or matrices \cite{Carlos18} (can be also found in Section \ref{sec4}). However, we do not use the doubled size for the sake of mathematical notation simplicity.  
} 

Furthermore, we form a finite set $\matc{A}$ with the size $J=L^M$. Every element in $\matc{A}$ forms a bijection to their corresponding codebook vector $\mathbf{x}_j\in\mathbb{R}^{M\times 1}$. According to {\em Definition 1}, \gls{vq} can be employed to form the mapping: $\mathbf{s} \longrightarrow j\longrightarrow\hat{\mathbf{x}}_j$.  For the \gls{mimo} detection, $\hat{\mathbf{x}}$ is the reconstructed version of $\mathbf{x}$, and thus it should follow the same distribution as $\mathbf{x}$ (i.e. the equal probability). Then, the average rate of codebook vectors becomes $R=M\log_2L$; and the compression rate of \gls{vq} is
\begin{equation}\label{eq05}
r=\frac{M\log_2L}{\hbar(\mathbf{s})}
\end{equation}
The average quantization loss given in {\em Definition 2} becomes
\begin{equation}\label{eq06}
\overline{\Delta}=\frac{1}{J}\sum_{j=0}^{J-1}\underset{\mathbf{s}}{\mathbb{E}}\Big(\min_j\Delta(\mathbf{s},\hat{\mathbf{s}}_j)\Big)
\end{equation}

{\color{sblue} Theoretically, the \gls{ann} architecture introduced in Section II-B can be straightforwardly employed for the \gls{mimo}-\gls{vq}. By means of the NN mapping, there will be $\bar{J}=L^M$ neurons integrated in the \gls{ann}, with each employing a binary output function. This might cause the well-known {\em curse of dimensionality} problem since neurons required to form the \gls{ann} grow exponentially with the number of transmit antennas and polynomially with the modulation order. For instance, there is a need of {\color{sblue} $\bar{J}=16^4$} neurons to detect \gls{mimo} signals sent by $M=4$ transmit antennas, with each modulated by $16$-QAM; and this is surely not a scalable solution.}

A considerable way to scale up the \gls{ann}-assisted \gls{mimo}-\gls{vq} is to employ the k-nearest neighbor (kNN) rule (see \cite{6313426}) to determine the output $\hat{z}_j$ \footnote{Note that the notation k here is not the dimension $\mathbb{R}^{K\times 1}$ defined in Section II-A. We take the notation directly from the name of kNN; and it would not be further used in other sections.}, i.e.,
\begin{equation}\label{eq07}
\hat{z}_j=\left\{\begin{array}{cl}
1,&\Delta(\mathbf{s}, \mathbf{w}_j)\in\{\text{k-min}\Delta(\mathbf{s}, \mathbf{w}_l),\forall l\}\\
0,&\text{otherwise}
\end{array}\right.
\end{equation}
where $\text{k-min}(\cdot)$ stands for the function to obtain the K minimums in a set. By such means, we can have $J\choose{\text{K}}$ codebook vectors through a combination of $\mathbf{w}_j, \forall j$. It is possible to further scale up \gls{mimo}-\gls{vq} through kNN with $\text{k}=1,...,K$, by means of which we have the capacity to form $\sum_{\text{k}=1}^\text{K} {J\choose{\text{k}}}$ codebook vectors \cite{10.1117/12.279444}.

{\color{sblue} Unlike conventional \gls{vq} techniques, \gls{mimo}-\gls{vq} can be associated with the supervised learning \cite{8761999}. At the training stage, the \gls{ann} exactly knows the mapping: $\mathbf{s}\longrightarrow j$, and thus it can iteratively adapt $\mathbf{w}_j$ to minimize the distortion $\Delta(\mathbf{s}, \mathbf{w}_j)$; and such largely reduces the learning complexity. Moreover, the knowledge of mapping $\mathbf{s}\longrightarrow j$ facilitates the combination of codebook vectors by which the number of required neurons can be significantly reduced.} Currently, there are two approaches for the codebook-vector combination:

\subsubsection{Cluster-level nearest neighbor} consider an \gls{ann} consisting of {\color{sblue} $\bar{J}=ML$} neurons. Neurons are divided into $M$ clusters indexed by $m$, corresponding to the $m^\mathrm{th}$ element in $\mathbf{x}$. Each cluster has $L$ neurons indexed by $l$, corresponding to the $l^\mathrm{th}$ state of an element in $\mathbf{x}$.  \gls{vq} forms a mapping $\mathbf{s}\longrightarrow(l_1, ..., l_M)$, with $l_m$ denoting the $l^\text{th}$ neuron within the $m^\text{th}$ cluster. The output is now indexed by $\hat{z}_{l,m}$, which is determined by the \gls{clnn} rule, $\forall m$
\begin{equation}\label{eq08}
\hat{z}_{l,m}=\left\{\begin{array}{cl}
1,&\Delta(\mathbf{s}, \mathbf{w}_{l,m})=\min\Delta(\mathbf{s}, \mathbf{w}_{\ell,m}), \forall\ell\\
0,&\text{otherwise}
\end{array}\right.
\end{equation}
By such means, one can use $(LM)$ neurons to represent $L^M$ codebook vectors. 

\subsubsection{Cluster-level kNN} consider an \gls{ann} consisting of {\color{sblue} $\bar{J}=M\log_2L$} neurons. Again, we can divide neurons into $M$ clusters. Each cluster now has $\alpha(=\log_2L)$ neurons. \gls{vq} maps $\mathbf{s}$ to the corresponding output $\hat{z}_{l,m}$ through the \gls{clknn} rule
\begin{equation}\label{eq09}
\hat{z}_{l,m}=\left\{\begin{array}{cl}
1,&\Delta(\mathbf{s}, \mathbf{w}_{l,m})\in\text{k-min}\Delta(\mathbf{s}, \mathbf{w}_{\ell,m}), \forall\ell\\
0,&\text{otherwise}
\end{array}\right.
\end{equation}
with $\text{k}\in\{1,...,L\}$. By such means, one can use $M\log_2L$ neurons to represent $L^M$ codebook vectors.

It is worth noting that both \gls{clnn} and \gls{clknn} are more suitable for supervised learning than unsupervised learning as far as the learning complexity is concerned. The ideas themselves are not novel. However, it is the first time in the literature that uses the \gls{vq} model to mathematically describe the \gls{ann}-assisted \gls{mimo} detection; and they lay the foundation for further investigations in other sections. 

\section{Fundamental Behaviors of \gls{mimo} Vector Quantization}\label{sec3}
The \gls{vq} concept presented in Section \ref{sec2} equips us with sufficient basis to develop a deeper understanding on fundamental behaviors of the \gls{ann}-assisted \gls{mimo}-\gls{vq}. 

\subsection{\gls{mimo}-\gls{vq} Compression Behavior}\label{sec3a}
In applications such as image, video, or speech compression, \gls{vq} has a controllable compression rate or quantization loss by managing the codebook size \cite{3776,Bovik:2005:HIV:1200338}. However, this is not the case for \gls{mimo}-\gls{vq}, as the codebook size is determined by the original signal. 
\begin{thm}\label{thm1}
Assuming the \gls{mimo} channel matrix $\mathbf{H}$ to be i.i.d.,  the compression rate of \gls{mimo}-\gls{vq} decreases linearly or more with the number of transmit antennas $M$. Moreover, we have $\lim_{M\rightarrow\infty}r=0$.
\end{thm}
\begin{proof}
Given the i.i.d. \gls{mimo} channel, we will have $\hbar(\breve{\mathbf{h}})=\sum_{n=0}^{N-1}\sum_{m=0}^{M-1}\hbar(H_{nm})=MN\sigma_H^2$, where $H_{nm}$ denotes the $(n, m)^\text{th}$ entry of $\mathbf{H}$ that has the entropy $\sigma_H^2$. The entropy of $\mathbf{s}$ is
\begin{IEEEeqnarray}{ll}
\hbar(\mathbf{s})&=\hbar(\mathbf{y})+\hbar(\mathbf{H})\label{eq10}\\
&\geq\hbar(\mathbf{x}|\mathbf{H})+\hbar(\mathbf{H})\label{eq11}\\
&\geq\hbar(\mathbf{x})+\hbar(\mathbf{H})\label{eq12}\\
&\geq M(\sigma_x^2+N\sigma_H^2)\label{eq13}
\end{IEEEeqnarray}
The equality in \eqref{eq11} holds in the noiseless case, and $\hbar(\mathbf{x}|\mathbf{H})=\hbar(\mathbf{x})$
due to $\mathbf{x}$ being independent from $\mathbf{H}$ in communication systems.
The information compression rate is therefore given by
\begin{IEEEeqnarray}{ll}
r=\frac{\hbar(\mathbf{x})}{\hbar(\mathbf{s})}&\leq\frac{\sigma_x^2}{\sigma_x^2+N\sigma_H^2}\label{eq14}\\
&\leq(1+M(\sigma_H^2)/(\sigma_x^2))^{-1}\label{eq15}
\end{IEEEeqnarray}
We can also use the \gls{vq} compression rate \eqref{eq05} to obtain
\begin{equation}\label{eq16}
r\leq(1+M(\sigma_H^2)/(\log_2L))^{-1}
\end{equation}
as far as the quantization of $\mathbf{x}$ is concerned. Both \eqref{eq15} and \eqref{eq16} show that the upper bound of the compression rate decreases linearly with $M$. For $M\rightarrow\infty$, it is trivial to justify that the upper bound tends to zero. {\em Theorem 1} is therefore proved. 
\end{proof}

The information-theoretical result in {\em Theorem 1} implies that \gls{mimo}-\gls{vq} introduces the rate distortion $(1-r)$, which grows at least linearly with $M$. Moreover, the rate distortion  approaches $100\%$ due to $\lim_{M\rightarrow\infty}r=0$; and such challenges the scalability of \gls{mimo}-\gls{vq}.

\subsection{MIMO-VQ Quantization Behavior}\label{sec3b}
We use the nearest-neighbor rule \eqref{eq03} to study the \gls{mimo}-\gls{vq} quantization behavior, appreciating its optimality as well as the simple mathematical form. Our analysis starts from the following hypothesis: 
\begin{prop}
Split the reconstructed vector $\hat{\mathbf{s}}_j$ into two parts: $\hat{\mathbf{s}}_j=[(\hat{\mathbf{s}}_j^{(1)})^T, (\hat{\mathbf{s}}_j^{(2)})^T]^T$, with $\hat{\mathbf{s}}_j^{(1)}\in\mathbb{R}^{N\times1}$, $\hat{\mathbf{s}}_j^{(2)}\in\mathbb{R}^{(MN)\times1}$. The quantization loss $\Delta(\mathbf{s}, \hat{\mathbf{s}}_j)$ can be decoupled into
\begin{equation}\label{eq17}
\Delta(\mathbf{s}, \hat{\mathbf{s}}_j)=\Delta(\mathbf{y}, \hat{\mathbf{s}}_j^{(1)})+\Delta(\breve{\mathbf{h}}, \hat{\mathbf{s}}_j^{(2)})
\end{equation}
\end{prop}
{\em Proposition 1} is true when we use the squared Euclidean-norm to measure the quantization loss; and in this case, \gls{vq} is least-squares optimal.

After training, the nearest-neighbor rule effectively leads to the forming of $J$ spheres, with each having their center at $\hat{\mathbf{s}}_j$ and radius $d_j$, i.e., 
\begin{equation}\label{eq18}
\Delta(\mathbf{s}, \hat{\mathbf{s}}_j)\leq d_j, \forall j.
\end{equation}
Moreover, we shall have
\begin{IEEEeqnarray}{rr}\label{eq19}
\Delta(\mathbf{y}, \hat{\mathbf{s}}_j^{(1)})\leq d_j^{(1)}&\\
\Delta(\breve{\mathbf{h}}, \hat{\mathbf{s}}_j^{(2)})\leq d_j^{(2)}\label{eq20}&
\end{IEEEeqnarray}
It is possible for $\mathbf{s}$ to fall into more than one sphere as described in \eqref{eq18}. In the \gls{ann} implementation, this problem can be solved through the use of softmax activation function \cite{8761999}. It is therefore reasonable to assume that $\mathbf{s}$ will only fall into one of the spheres. The \gls{mimo} detection will have errors when 
\begin{equation}\label{eq21}
\Delta(\mathbf{s}, \hat{\mathbf{s}}_j)\leq d_j, ~\text{for}~ \mathbf{s}\longrightarrow\mathbf{x}_i\longrightarrow\hat{\mathbf{s}}_i, i\neq j
\end{equation}
In order to prevent this case from happening in the noiseless context,  we shall have the following condition. 
\begin{thm}\label{thm2}
Assuming: a1) $\mathbf{H}_j\mathbf{x}_j=\mathbf{H}_i\mathbf{x}_i$, $\forall i\neq j$, a sufficient condition to prevent \eqref{eq21} from happening in the noiseless context is 
\begin{IEEEeqnarray}{rl}
\Delta(\breve{\mathbf{h}}_i, \hat{\mathbf{s}}_j^{(2)})&\geq d_j^{(2)}\label{eq22}\\
\Delta(\mathbf{H}_i\mathbf{x}_j, \hat{\mathbf{s}}_j^{(1)})&< d_j-d_j^{(2)}\label{eq23}
\end{IEEEeqnarray}
\end{thm}
\begin{proof}
In the noiseless case, given the channel realization $\mathbf{H}_j$, \eqref{eq17} can be written into
\begin{equation}\label{eq24}
\Delta(\mathbf{s}_j, \hat{\mathbf{s}}_j)=\Delta(\mathbf{H}_j\mathbf{x}_j, \hat{\mathbf{s}}_j^{(1)})+\Delta(\breve{\mathbf{h}}_j, \hat{\mathbf{s}}_j^{(2)})\leq d_j
\end{equation}
Considering $\mathbf{s}_i\longrightarrow\mathbf{x}_i$ with the channel realization $\mathbf{H}_i$, we have the quantization loss
\begin{equation}\label{eq25}
\Delta(\mathbf{s}_i, \hat{\mathbf{s}}_j)=\Delta(\mathbf{H}_i\mathbf{x}_i, \hat{\mathbf{s}}_j^{(1)})+\Delta(\breve{\mathbf{h}}_i, \hat{\mathbf{s}}_j^{(2)})> d_j
\end{equation}
Given the assumption a1), \eqref{eq25} becomes
\begin{equation}\label{eq26}
\Delta(\mathbf{s}_i, \hat{\mathbf{s}}_j)=\Delta(\mathbf{H}_j\mathbf{x}_j, \hat{\mathbf{s}}_j^{(1)})+\Delta(\breve{\mathbf{h}}_i, \hat{\mathbf{s}}_j^{(2)})> d_j
\end{equation}
To simultaneously fulfill \eqref{eq24} and \eqref{eq26} leads to
\begin{equation}\label{eq27}
\Delta(\breve{\mathbf{h}}_j, \hat{\mathbf{s}}_j^{(2)})\leq d_j-\Delta(\mathbf{H}_j\mathbf{x}_j, \hat{\mathbf{s}}_j^{(1)})<\Delta(\breve{\mathbf{h}}_i, \hat{\mathbf{s}}_j^{(2)}), \forall i\neq j
\end{equation}
\eqref{eq27} should hold for all possible realizations of $\mathbf{H}_j$ (or $\breve{\mathbf{h}}_j$). This condition can be guaranteed by replacing $\Delta(\breve{\mathbf{h}}_j, \hat{\mathbf{s}}_j^{(2)})$ with the upper bound \eqref{eq20}; and such leads to \eqref{eq22}. 

Furthermore, $\mathbf{H}_i$ can also be a possible realization of $\mathbf{H}_j$. In this case, we have
\begin{equation}\label{eq28}
\Delta(\mathbf{s}_j, \hat{\mathbf{s}}_j)=\Delta(\mathbf{H}_i\mathbf{x}_j, \hat{\mathbf{s}}_j^{(1)})+\Delta(\breve{\mathbf{h}}_i, \hat{\mathbf{s}}_j^{(2)})\leq d_j
\end{equation}
Applying the second inequality of \eqref{eq27} into \eqref{eq28} yeilds
\begin{IEEEeqnarray}{rl}\label{eq29}
\Delta(\mathbf{H}_i\mathbf{x}_j, \hat{\mathbf{s}}_j^{(1)})&<\Delta(\mathbf{H}_j\mathbf{x}_j, \hat{\mathbf{s}}_j^{(1)})\\
&<d_j-\Delta(\breve{\mathbf{h}}_j, \hat{\mathbf{s}}_j^{(2)})\label{eq30}
\end{IEEEeqnarray}
To guarantee the inequality \eqref{eq30}, $\forall j$, we shall replace $\Delta(\breve{\mathbf{h}}_j, \hat{\mathbf{s}}_j^{(2)})$ with their maximum $d_j^{(2)}$. The result \eqref{eq23} is therefore proved.
\end{proof}

{\color{sblue} {\em Theorem 2} indicates that \gls{mimo}-\gls{vq} actually consists of two parts of quantization, i.e., one for the \gls{csi} quantization, and the other for the received signal quantization. For the \gls{csi} quantization, each neuron statistically partitions \gls{csi} realizations into different groups according to the rule $\Delta(\breve{\mathbf{h}}, \hat{\mathbf{s}}_j^{(2)})\gtreqless d_j^{(2)}$, where the threshold $d_j^{(2)}$ is determined through supervised learning. Theoretically, the use of $\bar{J}$ neurons can result in $2^{\bar{J}}$ different \gls{csi} groups as a maximum. However, a channel realization fulfilling $\Delta(\breve{\mathbf{h}}, \hat{\mathbf{s}}_j^{(2)})\geq d_j^{(2)}$ does not normally fulfill $\Delta(\breve{\mathbf{h}}, \hat{\mathbf{s}}_i^{(2)})\geq d_i^{(2)}, \forall i\neq j$; and thus,  \gls{csi} realizations will only be partitioned into $\bar{J}$ groups, with each corresponding to a state for the transmitted signal partition. Thus, CSI quantization loss is inevitable when the number of channel realization sample is larger than the number of neurons on the output layer. As a result, an error floor is expected to occur at high \gls{snr} regime because of the learning inefficiency, i.e., channel ambiguity. More intuitively, the \gls{mimo}-\gls{vq} principle is illustrated in \figref{fig2}, 
\begin{figure}[t]
\begin{center}
\includegraphics[width=0.8\columnwidth]{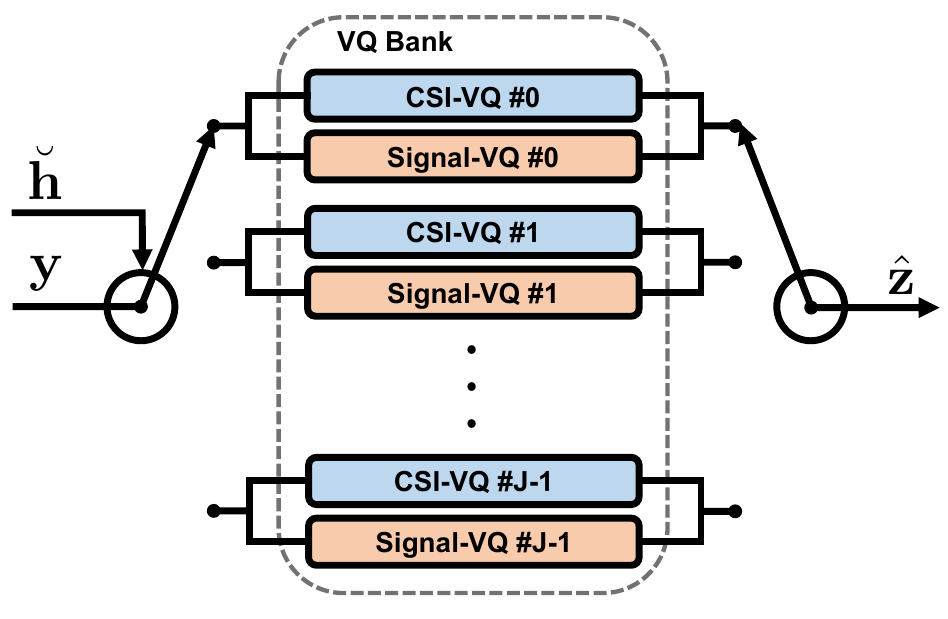}
\caption{Illustration of the \gls{mimo}-\gls{vq} principle. }\label{fig2}
\end{center}
\end{figure}
where the channel realization, $\breve{\mathbf{h}}$, together with received signal, $\mathbf{y}$, are mapped onto one out of the $\bar{J}$ states. It is worth noting that \gls{csi} quantization aims to remove the channel ambiguity defined in {\em Definition 2}. Meanwhile, the signal \gls{vq} performs a mapping $\mathbf{y}\longrightarrow \hat{\mathbf{z}}$, where $\mathbf{z}$ is the supervisory training target as described in \eqref{eq07}.

The major bottleneck here is again the {\em Curse of Dimensionality} problem \cite{Goodfellow-et-al-2016}, as CSI realizations are classified into $\bar{J}$ states which grow exponentially with the number of data streams $M$. When {\color{sblue}either} \gls{clnn} or \gls{clknn} (see Section II-C) is employed for the \gls{vq}, $\bar{J}$ grows linearly with $M$. However, such reduces the resolution of channel quantization and consequently increases the channel ambiguity in signal detection.}

\begin{figure*}[!t]
\centering
\includegraphics[width=7in]{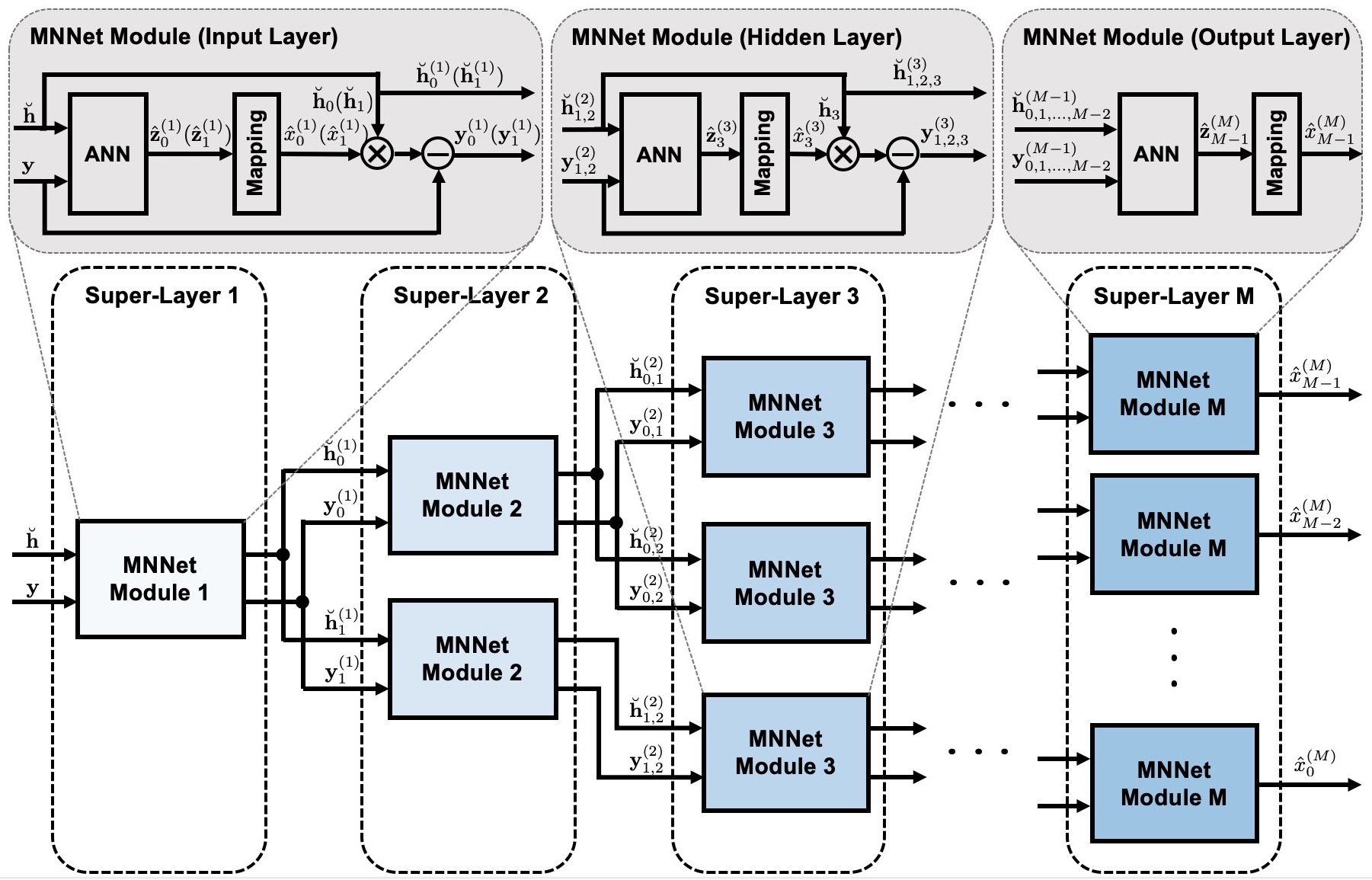}
\caption{Block diagram of the proposed MNNet.}
\label{fig3}
\end{figure*}

\subsection{Channel-Equalized Vector Quantization}\label{sec3c}
Section III-B shows that channel quantization scales down the \gls{ann}-assisted \gls{mimo} detectability. This straightforwardly motivates the channel-equalized \gls{vq} model as depicted in \figref{fig1}-(b). {\color{sblue} When channel equalizers ({\color{sblue}such as} \gls{zf} and \gls{lmmse}) are {\color{sblue}employed,  \gls{mimo}-\gls{vq} simply serves as a symbol de-mapper in Gaussian noise; and thus, there is no channel quantization loss. On the other hand, channel equalizers often require channel matrix inversion, which is of cubic complexity and not ready for parallel computing.
More critically,} the \gls{ann}-assisted receiver {\color{sblue}cannot} achieve the \gls{mlsd} performance as shown in \cite{8761999}, and there is {\color{sblue}lacking convincing} advantages of using machine learning for the \gls{mimo} detection. The only exceptional case is to use the \gls{mf} for the channel equalization, which supports parallel computing and low complexity. It is shown in \cite{8761999} that \gls{mimo}-\gls{vq} is able to improve the \gls{mf} detection performance by exploiting the sequence-detection gain. This is because the \gls{mf} channel equalizer does not remove the antenna correlation as much as the \gls{zf} and \gls{lmmse} algorithms, and leaves room for the \gls{ann} to exploit the residual sequence detection gain. On the other hand, the \gls{mf}-\gls{ann} approach has its performance far away from the \gls{mlsd} approach. Detailed computer simulation is provided in Section V.}

\section{{\color{sblue}The MNNet Approach for \gls{mimo}-\gls{vq}}}
\label{sec4}
{\color{sblue}
In this section, we introduce a novel \gls{mnn}-based deep learning approach, named MNNet, which combines modular neural network \cite{Happel94thedesign} with \gls{pic} algorithm \cite{563486}.

\subsection{MNNet Architecture}\label{sec4a}
The structure of the proposed MNNet is illustrated in \figref{fig3}. The entire network consists of $M$ cascade super-layers, and each layer has a similar structure that contains a group of MNNet modules instead of neurons. The MNNet modules on the same super-layer are identical, and those on different super-layers function differently. Besides, the MNNet modules on the front $M-1$ super-layers share similar structures, which consist of three basic components: An \gls{ann}, a codebook mapping function and an interference cancellation unit. The $M^\mathrm{th}$ super-layer does not have the interference cancellation unit because the output of its mapping function is the information-bearing symbols.

The input to the first super-layer is a concatenation of the \gls{csi} vector $\breve{\mathbf{h}}$ and the received signal vector $\mathbf{y}$. It is perhaps worth noting that communication signals are normally considered as complex-valued symbols, but most of the existing deep learning algorithms are based on real-valued operations. To facilitate the learning and communication procedure,  it is usual practice to convert complex signals to their real signal equivalent version using \eqref{eq31} (see \cite{8642915,DBLP:journals/corr/abs-1812-10044,8646357,8054694,Carlos18}). For instance, a $(K) \times (1)$ complex-valued input vector $\mathbf{s}$ is converted into a $(2K) \times (1)$ real-valued vector $\mathbf{s}_\text{real}$ by concatenating its real and imaginary parts, which is given by
\begin{equation}\label{eq31}
\mathbf{s}_\text{real} = \begin{bmatrix}
~\Re(\mathbf{s})~\\ 
~\Im(\mathbf{s})~
\end{bmatrix}
\end{equation}
The output of \gls{ann} is the estimate of entire codebook vector $\hat{\mathbf{z}}^{(1)}$ as we introduced in Section \ref{sec2b}\&C, where the superscript $[\cdot]^{(1)}$ stands for the number of super-layer. In order to simplify the network structure, only the first two transmitted symbols $\hat{x}_0^{(1)}(\hat{x}_1^{(1)})$ are mapped back from the estimated codebook vector $\hat{\mathbf{z}}_0^{(1)}(\hat{\mathbf{z}}_1^{(1)})$ by applying the pre-defined bijection mapping functions (e.g. \gls{nn}, \gls{clnn} or \gls{clknn}). Denote $\breve{\mathbf{h}}_m$ to the channel coefficients between the $m^\mathrm{th}$ transmit antenna to the \gls{mimo} receiver, the interference cancellation procedure can be described as 
\begin{equation}\label{eq32}
\mathbf{y}_m^{(1)}=\mathbf{y}-\breve{\mathbf{h}}_m\hat{x}_m^{(1)},~_{m\in\left\{0,1\right\}}
\end{equation}
which, together with part of the \gls{csi} vector $\breve{\mathbf{h}}_m^{(1)}$ serves as the input to the second super-layer, i.e., $\mathbf{s}_m^{(1)}=[\mathbf{y}_m^{(1)T},\breve{\mathbf{h}}_m^{(1)T}]^T,~_{m\in\left\{0,1\right\}}$, where $\breve{\mathbf{h}}_m^{(1)}$ is obtained by removing $\breve{\mathbf{h}}_m$ from $\breve{\mathbf{h}}$. Therefore, the second super-layer will consist of two parallel MNNet modules where the outputs from the previous super-layer are processed separately. In order to obtain an estimate of every transmitted symbol $\hat{x}_m,~_{0\leq m\leq M-1}$, the top branch on each super-layer needs to produce two outputs as depicted in \figref{fig3}. Furthermore, the MNNet modules on the front $M-1$ super-layers share similar structure except for the size of \gls{ann} input and output. Specifically, the size of input to the $m^\mathrm{th}$ super-layer is $2N(M-m+2)$, and output is $2^{(M-m+1)\cdot \log_2L}$ when \gls{nn} mapping rule is considered. By repeating this process for $M$ times, the output of the $M^\mathrm{th}$ super-layer are the final estimates of transmitted symbols.

\subsection{Scalability and Computational Complexity}
\label{sec4b}
The information-theoretical result in \eqref{eq15} and \eqref{eq16} implies that \gls{ann}-assisted \gls{mimo}-\gls{vq} introduces the rate distortion, which grows at least linearly with the number of transmit antennas $M$. The proposed MNNet scales up the \gls{mimo}-\gls{vq} solution by reducing the number of data streams on a layer-by-layer basis. The upper bound of the compression rate increases along the feed-forward propagation. On the other hand, MNNet involves a large number of modules. 
\begin{prop}\label{prop2}
Concerning every super layer to be associated with the PIC algorithm, the MNNet involves in total $M(M+1)/2$ modules to conduct the \gls{mimo} signal detection.
\end{prop}
\begin{proof}
As shown in \figref{fig3}, there are $M$ super layers forming the MNNet. On each $m^\mathrm{th}$ super layer ($m=1,...,M$),  MNNet module eliminates $(m-1)$ interferences, and then perform the \gls{vq} procedure. According to the \gls{pic} principle, interferences eliminated at MNNet modules have to be different and in order. Hence, the $m^\mathrm{th}$ super layer needs $P(M, m-1)$ modules, where $P(,)$ denotes the permutation. However, since every transmitted symbol are independently drawn from a finite-alphabet set with equal probability. It is trivial to justify that the average error probability of each reconstructed transmitted symbol is equal. Therefore, the total number of modules employed in MNNet can calculated by 

\begin{IEEEeqnarray}{ll}
\mathcal{Q}(M)&= \sum_{m=1}^{M}m \nonumber \\
&= M(M+1)/2
\end{IEEEeqnarray}
Proposition 2 is therefore proved.
\end{proof}

Since the MNNet is well modularized, the learning procedure does not need to be applied on the entire network as a whole, but rather at the modular level. Such a strategy largely improves the computational efficiency at the ANN training stage. The computational complexity required for MNNet training per module per iteration is approximately $\mathcal{O}(bNM)$, where $b$ is the size of mini-batch; and $\mathcal{O}(NM^2)$ per detection. To put this in perspective, DetNet has a complexity of $\mathcal{O}(dN^2)$, where $d$ is the number of iterative layers. OAMPNet has a higher complexity of $\mathcal{O}(dN^3)$ dominated by the matrix inversion. The \gls{lmmse} algorithm similarly requires a matrix inversion, resulting in a complexity of $\mathcal{O}(N^3)$.

}

\section{{\color{sblue}Simulation} Results and Analysis}
\label{sec5}
{\color{sblue} This section presents the experimental results and related analysis. The data sets and experimental setting are introduced at the beginning, followed by a brief introduction of the competing algorithms. After that, a comprehensive performance evaluation is given, which demonstrates the performance of the proposed MNNet approach.}

\subsection{Data Sets and Experimental Setting}\label{sec5a}
In traditional deep learning applications such as image processing and speech recognition, the performance of different algorithms and models are evaluated under common benchmarks and open data sets (e.g. MNIST, LSUN) \cite{888}. It is a different story in wireless communication domain since we are dealing with artificially manufactured data that can be accurately generated. Therefore, we would like to define the data generation routines instead of giving specific data sets.

{\color{sblue} As far as the supervised learning is concerned, every training sample is a mainstay containing an input object and a supervisory output. According to the system model in \eqref{eq01}, the transmitted signal $\mathbf{x}$ is an $(M)\times(1)$ vector with each symbol drawn from a finite-alphabet set $\matc{A}$ consisting of $L$ elements. There are a total number of $J=L^M$ possible combinations, and $\mathbf{x}_j$ denotes the $j^\text{th}$ combination. By applying the \gls{nn} rule, the corresponding supervisory representation of $\mathbf{x}_j$ can be expressed by
\begin{equation}\label{eq34}
\mathbf{z}_j=[\mathbf{1}(\mathbf{x}_j=\mathbf{x}_0),\mathbf{1}(\mathbf{x}_j=\mathbf{x}_1),\cdots,\mathbf{1}(\mathbf{x}_j=\mathbf{x}_{J-1}) ]^T
\end{equation}
where $\mathbf{1}(\cdot)$ is the indicator function. It is easy to find that $\mathbf{z}_j$ is a $(J)\times(1)$ one-hot vector with only one element equals to one and others are zero. Meanwhile, other codeword mapping approaches including \gls{clnn} and \gls{clknn} can be implemented through similar method as described in \eqref{eq08} or \eqref{eq09}. Due to the channel randomness, each transmitted signal can yield to multiple received signals even in the noiseless case. Denote $\matc{B}_j$ to the feasible set which contains all possible received signal vectors when $\mathbf{x}_j$ is transmitted, $\mathbf{y}_{i,j}\in \matc{B}_j$ to the $i^\text{th}$ element of $\matc{B}$ and $\breve{\mathbf{h}}_{i,j}$ to the corresponding channel realization. Then, the pairwise training sample can be described as
\begin{equation}\label{eq35}
\left \{ [\breve{\mathbf{h}}_{i,j}^T,\mathbf{y}_{i,j}^T]^T,\mathbf{z}_j \right \}
\end{equation}
and the goal of neural network training is to minimize the following objective function
\begin{equation}\label{eq36}
\varphi ^*= \underset{\varphi}{\arg \min} ~\mathcal{L}\left (\mathbf{z}_j,\hat{\mathbf{z}}_j  \right )
\end{equation}
by adjusting the trainable parameters $\varphi = \left\{\mathbf{W},\mathbf{b} \right\}$, where $\mathcal{L}(\cdot)$ denotes to loss function, $\mathbf{W}$ and $\mathbf{b}$ to the weight and bias, and $\hat{\mathbf{z}}_j$ to the \gls{ann} reconstructed vector. The most popular method for updating parameters $\varphi$ is backpropagation together with stochastic gradient descent \cite{Rumelhart:1995:BBT:201784.201785}, which start with some randomly initialized values and iteratively converge to an optimum point. Furthermore, there are many adaptive learning algorithms which dramatically improve the convergence performance in neural network training procedure \cite{Riedmiller1994AdvancedSL}. In this paper, AdaBound algorithm (see \cite{luo2018adaptive} for detailed description) is applied for performance evaluation

All the experiments are run on a Dell PowerEdge R730 2x 8-Core E5-2667v4 Server, and implemented in MATLAB.

\subsection{{\color{sblue}Baseline Algorithms for Performance Comparison} }\label{sec5b}
In our experiments, the following algorithms are employed for performance comparison:
\begin{itemize}
\item \textbf{ZF:} Linear detector that applies the channel pseudo-inverse to restore the signal \cite{5534810}.
\item \textbf{LMMSE:} Linear detector that applies the SNR-regularized channel pseudo-inverse to restore the signal \cite{DBLP:journals/corr/abs-1102-1462}.
\item \textbf{MF:} Linear detector that has the lowest computational complexity among all \gls{mimo} detectors \cite{8254935}.
\item \textbf{MLSD:} The optimum detection algorithm that requires an exhaustive search.
\item \textbf{DetNet:} A deep learning approach introduced in \cite{8642915}, which applies a deep unfolding approach that transforms a computationally intractable probabilistic model into a set of cascaded deep neural networks.
\item \textbf{OAMPNet:} A deep learning approach introduced in \cite{8646357}, which incorporates deep learning into the orthogonal AMP algorithm.
\item \textbf{IW-SOAV:} An iterative detection algorithm for massive overloaded \gls{mimo} system based on iterative weighted sum-of-absolute value optimization \cite{7760475}.
\item \textbf{TPG-Detector:}  A deep learning approach introduced in \cite{DBLP:journals/corr/abs-1812-10044}, which is designed specifically for massive overloaded \gls{mimo} system.
\end{itemize}

\subsection{Simulations and Performance Evaluation}\label{sec5c}
Our computer simulations are structured into five experiments. In experiment 1, we investigate the performance of different codebook combination approaches. In experiment 2, the impact of the size of channel realization set on the detection performance of \gls{mimo}-\gls{vq} model is studied. In experiment 3 and 4, we evaluate the detection performance of both channel-equalized \gls{vq} model and \gls{mimo}-\gls{vq} model. In experiment 5, a comprehensive performance evaluation is provided for the proposed MNNet algorithm. The key metric utilized for performance evaluation is the average \gls{ber} over sufficient Monte-Carlo trails of block Rayleigh fading channels. The \gls{snr} is defined as the average received information bit-energy to noise ratio per receive antenna (i.e. $E_b/N_0$).

\begin{figure}[t!]
\centering
\includegraphics[width=3.5in]{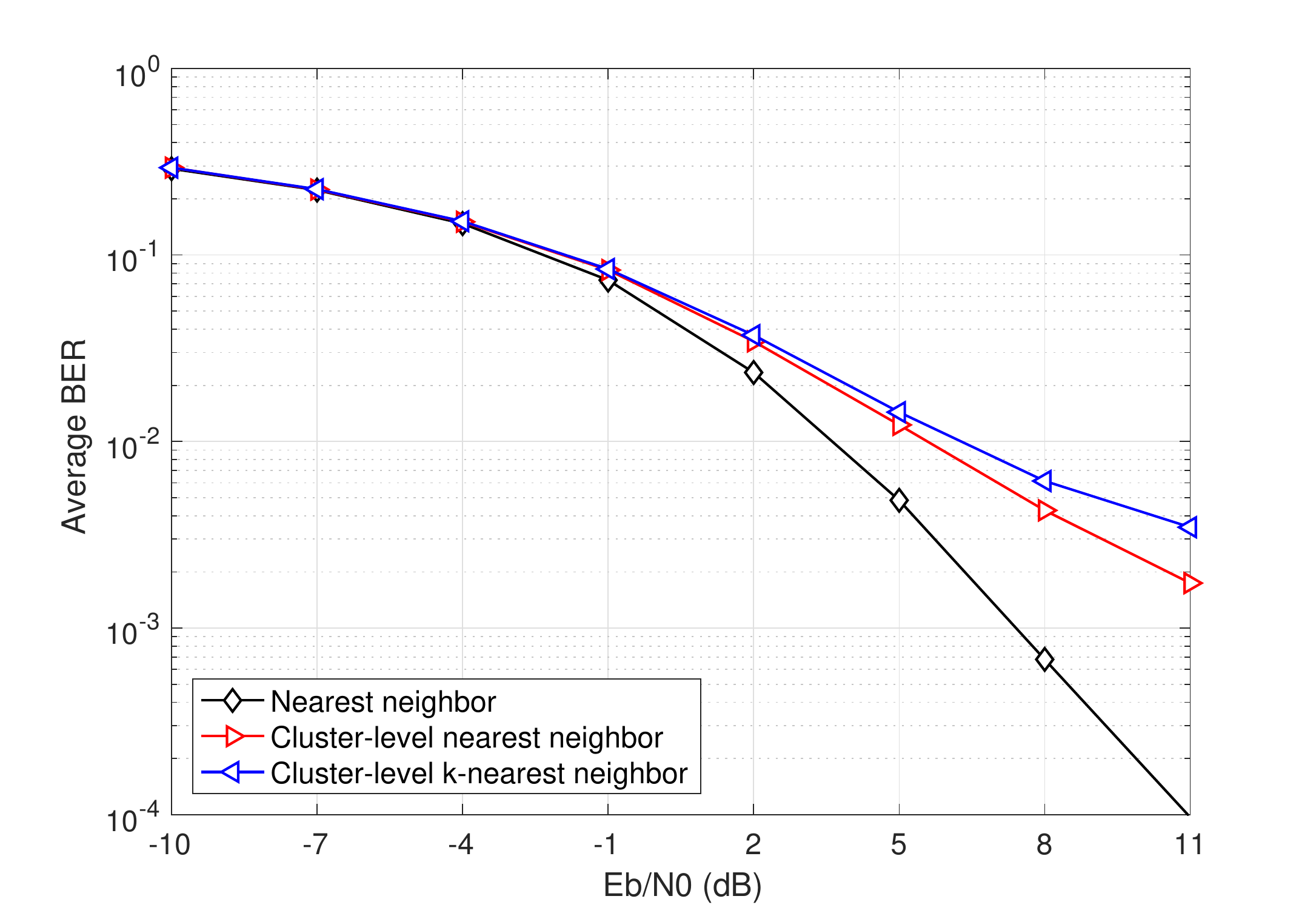}
\caption{\gls{ber} as a function of $E_b/N_0$ for \gls{mimo}-\gls{vq} model with different codebook combination approaches including \gls{nn}, \gls{clnn} and \gls{clknn} in uncoded 4-by-8 \gls{mimo} system with QPSK modulation.}
\label{fig4}
\end{figure}

\begin{table}
\center
\caption{Layout of the \gls{ann}}\label{T1}
\renewcommand{\arraystretch}{1}
\begin{tabular}{l|l}
\hline
\textbf{Layer} &\textbf{Output dimension} \\ \hline
Input  & $2K$ \\ \hline
Dense + ReLU & $1024$   \\ \hline
Dense + ReLU & $512$   \\ \hline
Dense + ReLU & $256$   \\ \hline
\begin{tabular}[c]{@{}l@{}}1) \gls{nn}: Dense + softmax\\ 2) \gls{clnn}: Dense + cluster-level softmax\\ 3) \gls{clknn}: Dense + sigmoid\end{tabular} & \begin{tabular}[c]{@{}l@{}}1) $L^M$\\ 2) $ML$\\ 3) $M\log_2L$\end{tabular} \\ \hline
\end{tabular}
\end{table}

\subsubsection*{Experiment 1} The aim of this experiment is to investigate the performance of different codebook combination approaches including \gls{nn}, \gls{clnn} and \gls{clknn}. The system environment is set to be an uncoded 4-by-8 \gls{mimo} network with QPSK modulation scheme. Moreover, Table \ref{T1} provides a detailed layout of \gls{ann} that applied for signal detection. The \gls{ann} training is operated at $E_b/N_0=5$ dB with a mini-batch size of 500; as the above configurations are found to provide the best detection performance.

\figref{fig4} illustrates the average \gls{ber} performance of the \gls{mimo}-\gls{vq} model with different codebook combination approaches. It is observed that \gls{nn} significantly outperforms the other two approaches throughout the whole \gls{snr} range. For \gls{clnn}, the detection performance decreased approximately 4.7 dB at \gls{ber} of $10^{-3}$; and for \gls{clknn}, the gap is around 5 dB. The reason of performance degradation is because of the channel learning inefficiency, i.e., channel ambiguity. Moreover, \gls{clnn} outperforms \gls{clknn} at high \gls{snr} range thanks to the use of more neurons on the output layer of \gls{ann} (i.e. $ML>M\log_2L$). The above phenomenons coincide with our conclusions made in Section \ref{sec2b}\&C that \gls{nn} offers the best performance among all codebook combination approaches. Both \gls{clnn} and \gls{clknn} make a trade-off between computational complexity and performance. In order to demonstrate the best performance, \gls{nn} is applied in the following experiments.

\subsubsection*{Experiment 2} The aim of this experiment is to investigate the impact of the size of channel realization set on the detection performance of \gls{mimo}-\gls{vq} model. The system environment and layout of \gls{ann} remain unchanged as we introduced in Experiment 1. Besides, the \gls{ann} is trained at $E_b/N_0=8$ dB with a mini-batch size of 500.

\begin{figure}[t!]
\centering
\includegraphics[width=3.5in]{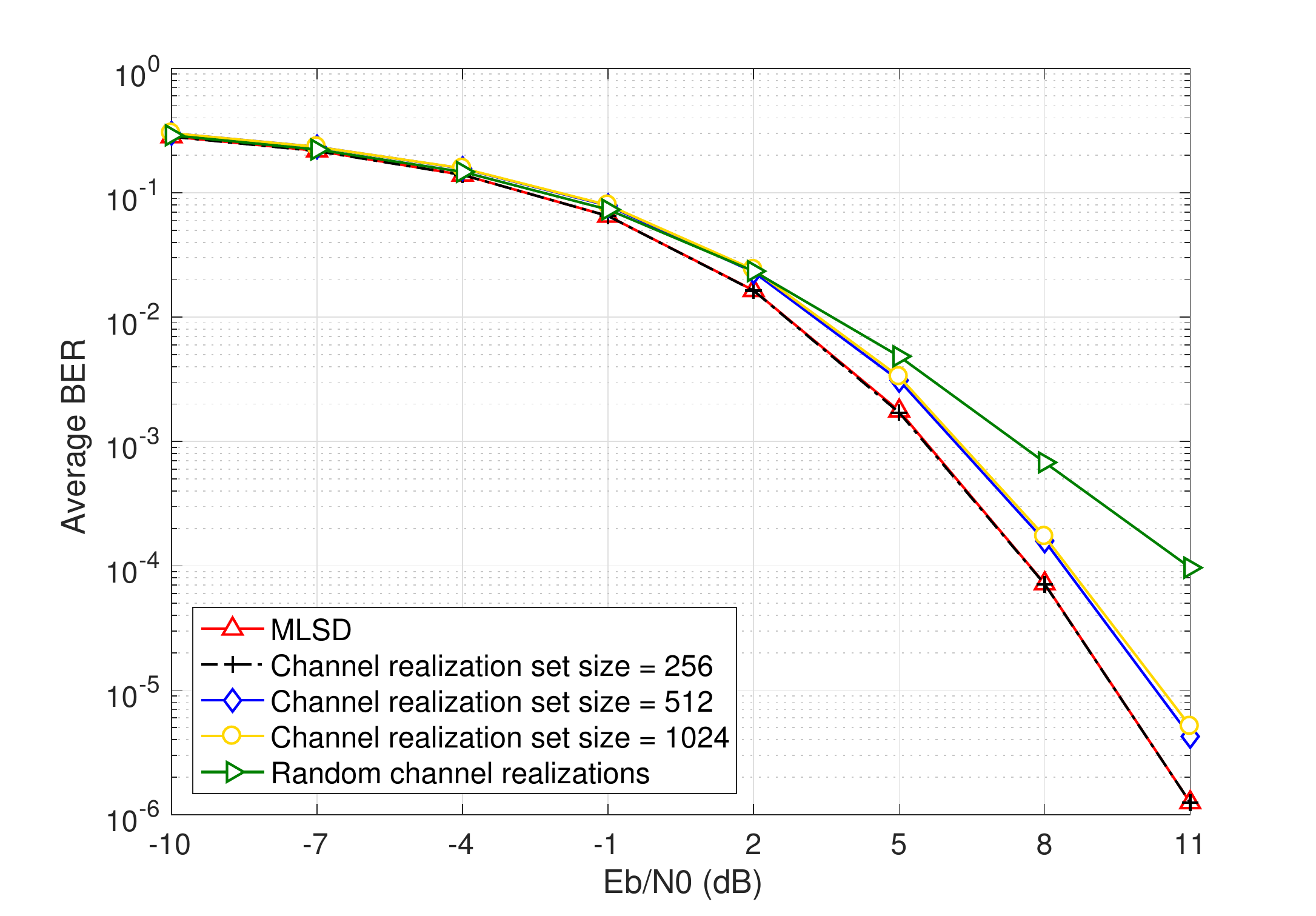}
\caption{\gls{ber} as a function of $E_b/N_0$ for \gls{mimo}-\gls{vq} model with different sizes of channel realization set in uncoded 4-by-8 \gls{mimo} system with QPSK modulation.}
\label{fig5}
\end{figure}

\figref{fig5} illustrates the average \gls{ber} performance of the \gls{mimo}-\gls{vq} model with different sizes of channel realization set. The channel data set is obtained by randomly generate a certain number of unique channel matrices and is utilized for both \gls{ann} training and testing phases. Since the environment is set to be an uncoded 4-by-8 \gls{mimo} with QPSK modulation, the number of neurons on \gls{ann} output layer is 256 (i.e. $\bar{J}=4^4$ as \gls{nn} is considered). Therefore, we consider 4 different cases with the size of the channel realization set varies from 256 to infinite \footnote{Here infinite set stands for the case that \gls{mimo} channel matrix is randomly generated in each training and testing iteration.}. The baseline for performance comparison is the optimum \gls{mlsd}. It is observed that \gls{ann}-assisted \gls{mimo} receiver achieves optimum performance when the size of the set is 256, as the size of channel realizations matches the number of neurons on \gls{ann} output layer. By increasing the size to 512, the detection accuracy slightly degraded around 1 dB at high \gls{snr} range. Further increasing the size of channel realizations to 1024 does not make significant performance difference, as the gap between 512 and 1024 is almost negligible. The situation becomes different when the channel data set is not pre-defined (i.e. randomly generated every time). The average \gls{ber} performance degraded about 2.5 dB at \gls{ber} of $10^{-4}$. The above phenomenons coincide with our conclusion in Section \ref{sec3b} that channel quantization level bounds the signal quantization performance. When the size of the channel data set increases, the channel learning inefficient (i.e. channel ambiguity) causes the loss of detection accuracy particularly at high \gls{snr} range.

\subsubsection*{Experiment 3} This aim of this experiment is to demonstrate the performance of channel-equalized \gls{vq} model. Three commonly used channel equalizers are considered (e.g. \gls{zf}, \gls{lmmse} and \gls{mf}). The layout of \gls{ann} is slight different since the input to \gls{ann} becomes the channel equalized signal block with a dimension of input equals to $2\times M$. Meanwhile, the output of \gls{ann} remains unchanged as introduced in Table I. Besides, the \gls{mimo} channel is considered as slow fading which varies every three transmission slots. The \gls{ann} is trained at $E_b/N_0=10$ dB with a mini-batch size of 500.

\begin{figure}[t!]
\centering
\includegraphics[width=3.5in]{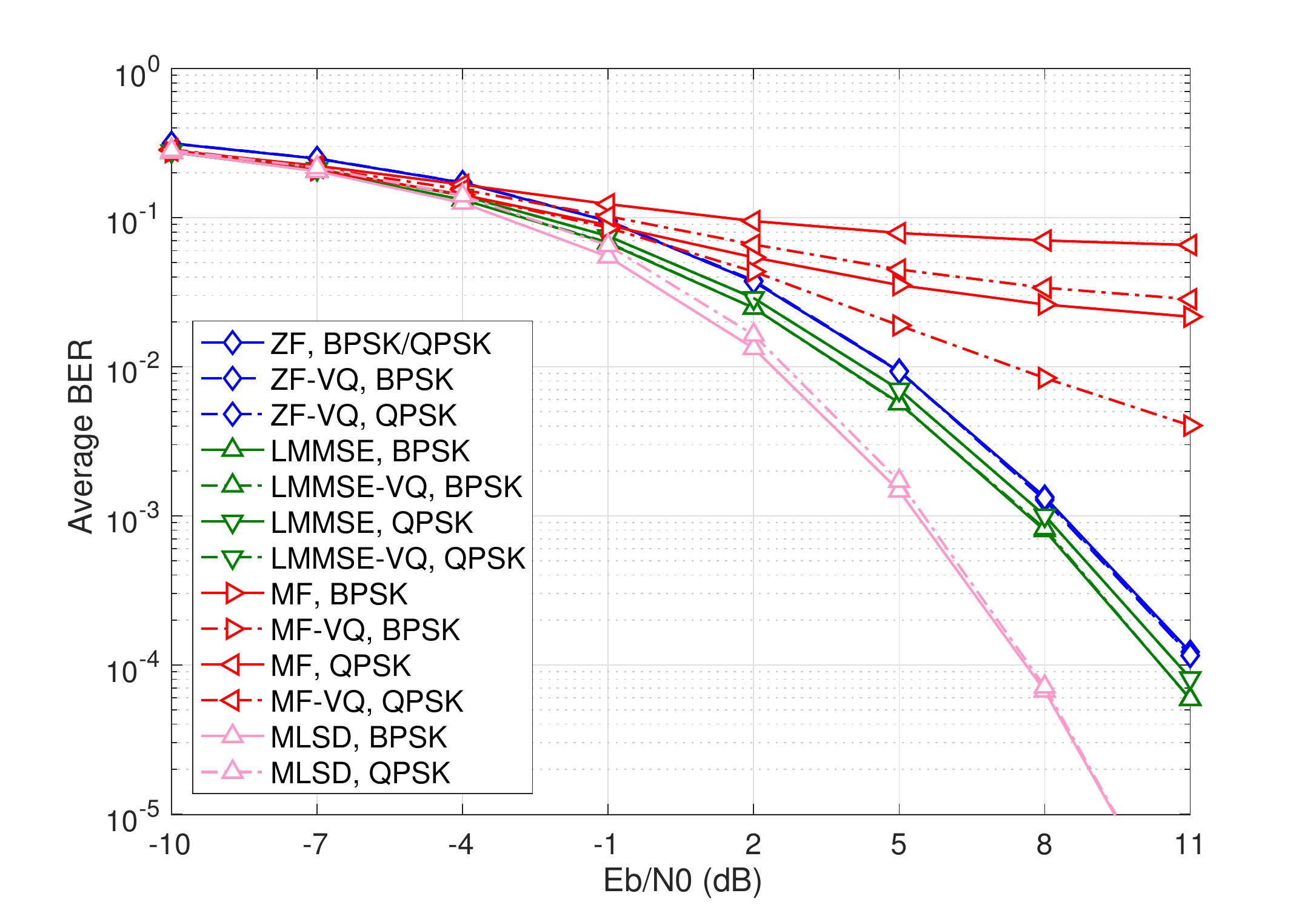}
\caption{\gls{ber} as a function of $E_b/N_0$ for channel-equalized \gls{vq} model in uncoded 4-by-8 MIMO system with BPSK/QPSK modulation.}
\label{fig6}
\end{figure}

\figref{fig6} illustrates the average \gls{ber} performance of channel-equalized \gls{vq} model in uncoded 4-by-8 \gls{mimo} system with both BPSK and QPSK modulation schemes. The baseline for performance comparison is the conventional channel equalizers.  It is observed that \gls{ann} significantly improves the detection performance of \gls{mf} equalizer at high \gls{snr} range. The deep learning gain is around 7 dB for BPSK and 10 dB for QPSK. This phenomenon indicates our conclusion in Section \ref{sec3c} that \gls{ann} is able to improve \gls{mf} equalizer by exploiting more sequence detection gain. However, for both \gls{zf} and \gls{lmmse} equalizers, \gls{ann} does not offer any performance improvement, as the \gls{ann} serves as a simple \gls{awgn} de-mapper. Moreover, all of these approaches have their performances far away from the optimum \gls{mlsd} in both BPSK and QPSK cases.

\begin{figure}[t!]
\centering
\includegraphics[width=3.5in]{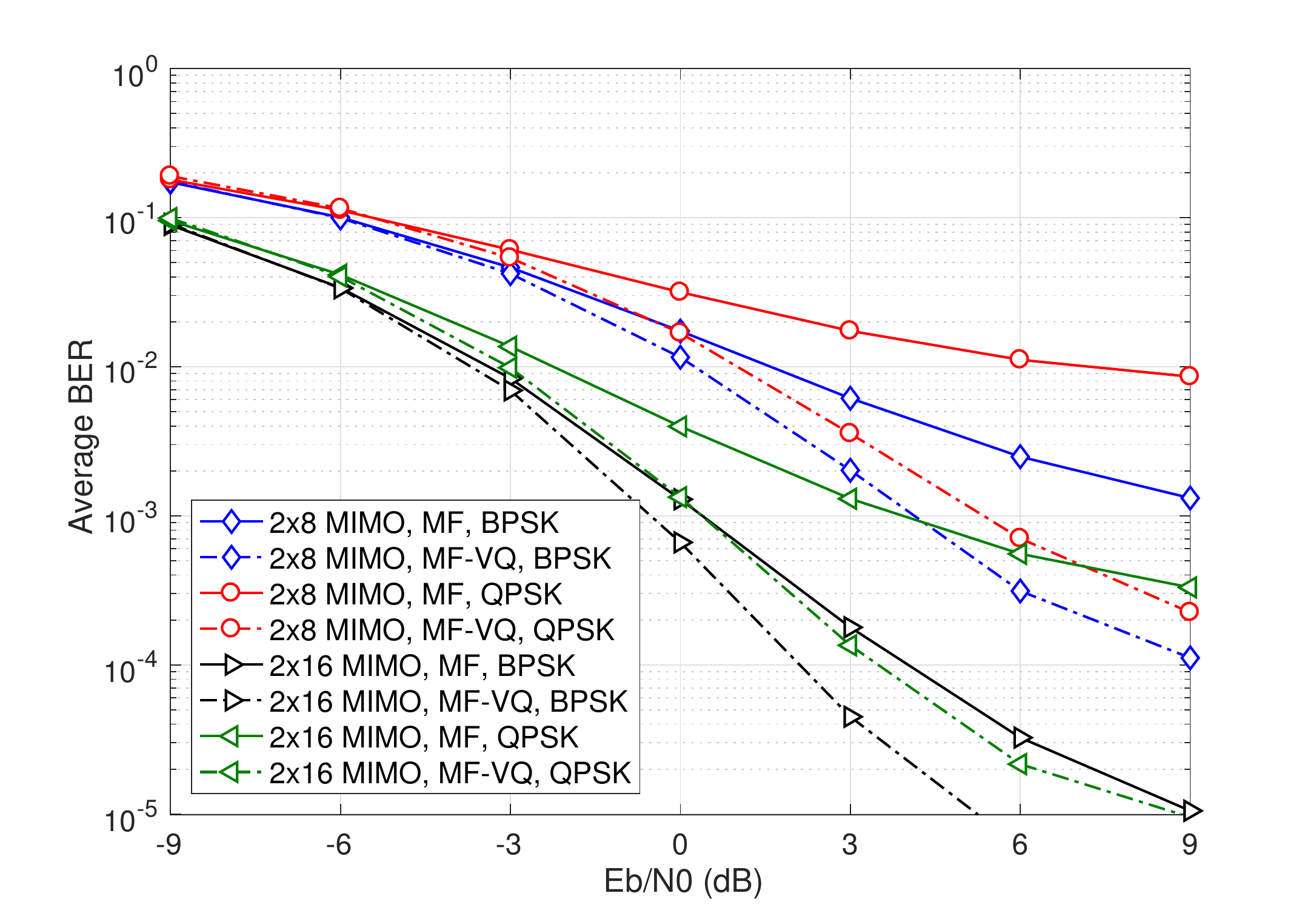}
\caption{\gls{ber} as a function of $E_b/N_0$ for \gls{mf}-\gls{vq} model in uncoded MIMO system with BPSK/QPSK modulation.}
\label{fig7}
\end{figure}

\figref{fig7} illustrates a more detailed performance evaluation for \gls{mf}-\gls{vq} model. The performance of \gls{ann} is compared with the conventional \gls{mf} equalizer. It is observed that \gls{mf}-\gls{vq} model largely improves the detection performance in all scenarios. For 2-by-8 \gls{mimo}, the sequence detection gain is approximately 3 dB for BPSK and 5 dB for QPSK; and for 2-by-16 \gls{mimo}, the gain is more than 2 dB for BPSK and 4 dB for QPSK at high \gls{snr} range.

\subsubsection*{Experiment 4} The aim of this experiment is to demonstrate the performance of \gls{mimo}-\gls{vq} model with different modulation schemes. For BPSK and QPSK modulation, \gls{mimo}-\gls{vq} is trained at $E_b/N_0=5$ dB with a mini-batch size of 500; for 8-PSK and 16-QAM modulation, the training $E_b/N_0$ is set at 8 dB. The reason for training ANN at different SNR points for different modulation schemes is because we found different SNR points will result in different detection performances in ANN-assisted MIMO communications. Specifically, training at lower SNRs can obtain an \gls{ann} receiver only works well at low-SNR range, and vice versa. In order to achieve the best performance throughout the whole SNR range, we tested a few SNR points and the above settings are observed to achieve the best end-to-end performance among others.

\begin{figure}[t!]
\centering
\includegraphics[width=3.5in]{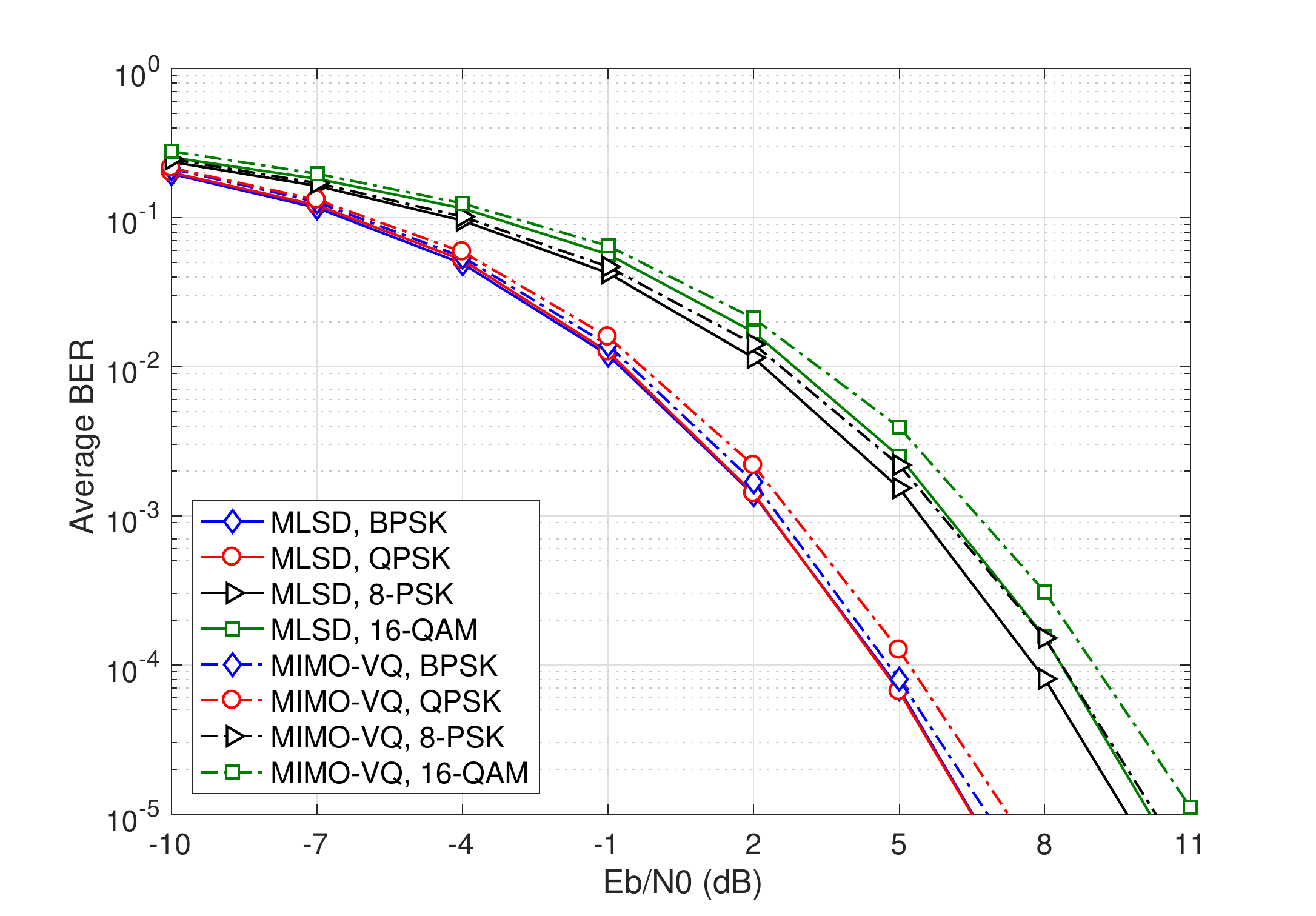}
\caption{\gls{ber} as a function of $E_b/N_0$ for \gls{mimo}-\gls{vq} model in uncoded 2-by-8 \gls{mimo} system with multiple modulation schemes.}
\label{fig8}
\end{figure}

\figref{fig8} illustrates the average \gls{ber} performance of \gls{mimo}-\gls{vq} model in uncoded 2-by-8 \gls{mimo} system with multiple modulation schemes. The baseline for performance comparison is the optimum \gls{mlsd}. It is observed that \gls{mimo}-\gls{vq} model is able to achieve near-optimum performance in all cases. Specifically, the performance gap for BPSK is almost negligible; for QPSK and 8-PSK, the gap is around 0.8 dB at \gls{ber} of $10^{-5}$; for 16-QAM, the gap increases to around 1.1 dB at \gls{ber} of $10^{-5}$. Again, the performance gap is introduced by channel learning inefficient as we discussed in Experiment 2. Moreover, it is observed that channel quantization loss is more emergent for higher modulation schemes. This is perhaps due to the denser constellation points in higher-order modulation schemes (e.g. 16-QAM).

\subsubsection*{Experiment 5} The aim of this experiment is to examine the detection performance of the proposed MNNet algorithm. Here we consider three different system models including overloaded ($M>N$), under loaded ($M<N$) and critically loaded ($M=N$) \gls{mimo} systems. Besides, several deep learning based \gls{mimo} detection algorithms are applied for performance comparison including DetNet, OAMPNet and TPG-Detector. Specifically, DetNet is implemented in 20 layers, and OAMPNet is implemented in 10 layers with 2 trainable parameters and a channel inversion per layer.

\begin{figure}[t!]
\centering
\includegraphics[width=3.5in]{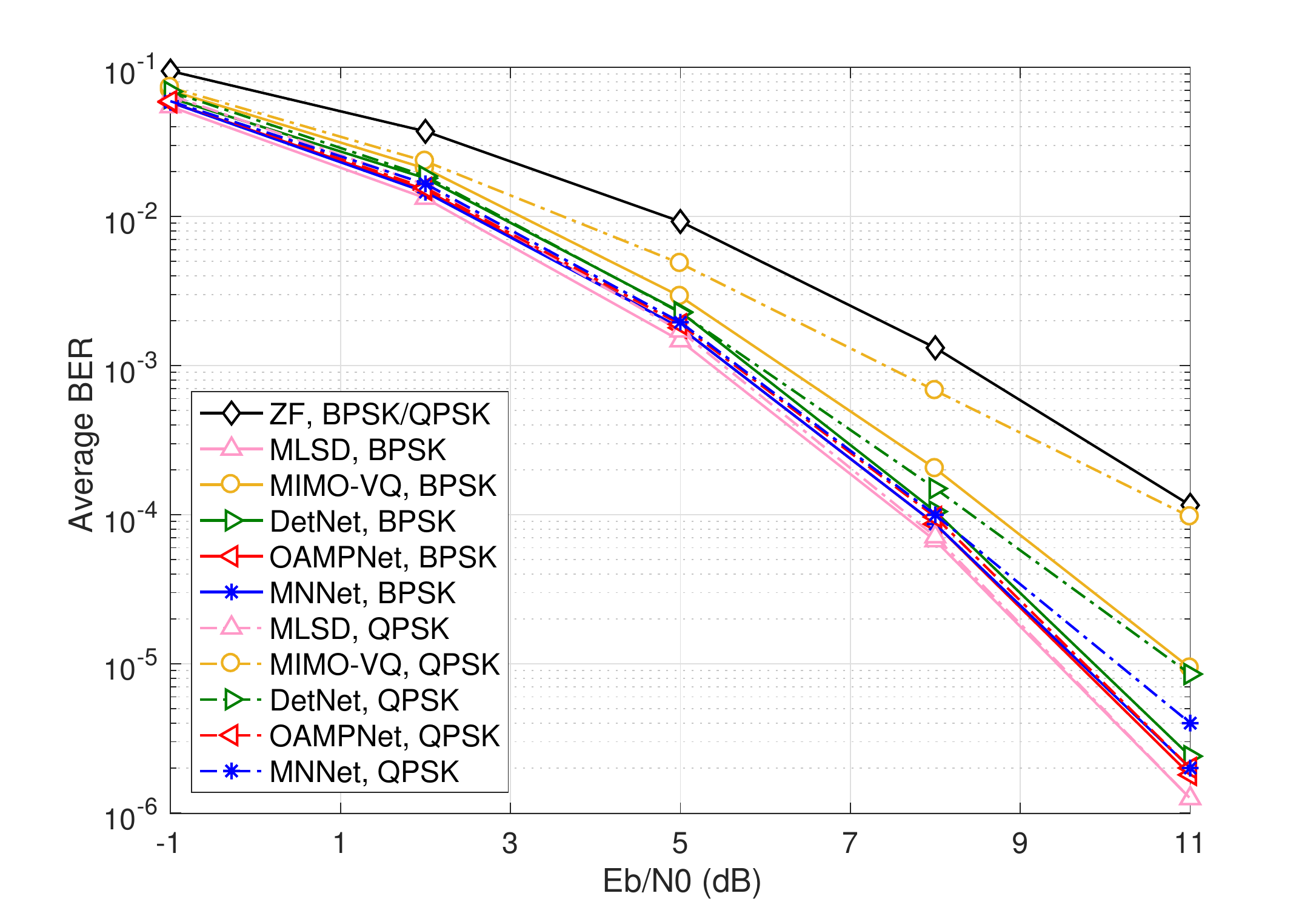}
\caption{\gls{ber} as a function of $E_b/N_0$ for different detection algorithms in uncoded 4-by-8 \gls{mimo} system with BPSK/QPSK modulation.}
\label{fig9}
\end{figure}

\figref{fig9} illustrates the average \gls{ber} performance of different detection algorithms in uncoded 4-by-8 \gls{mimo} system with BPSK/QPSK modulation. It is observed that all of the deep learning algorithms outperform the conventional \gls{zf} algorithm throughout the whole \gls{snr} range. Specifically, the \gls{mimo}-\gls{vq} model has its performance quickly moved away from the optimum due to the increasing number of transmit antennas. The performance degraded around 1 dB for BPSK and 3.2 dB for QPSK. DetNet achieves a good performance on BPSK, but its performance gap with \gls{mlsd} increases to 1.7 dB when we move to QPSK. Meanwhile, OAMPNet and MNNet approaches are both very close to optimum for BPSK modulation over the whole \gls{snr} range. The performance of MNNet slightly decreased as we move to QPSK; the performance gap to \gls{mlsd} is around 0.4 dB at \gls{ber} of $10^{-5}$. Despite, the proposed MNNet does not require matrix inversion which makes it more advantage in terms of computational complexity.

\begin{figure}[t!]
\centering
\includegraphics[width=3.5in]{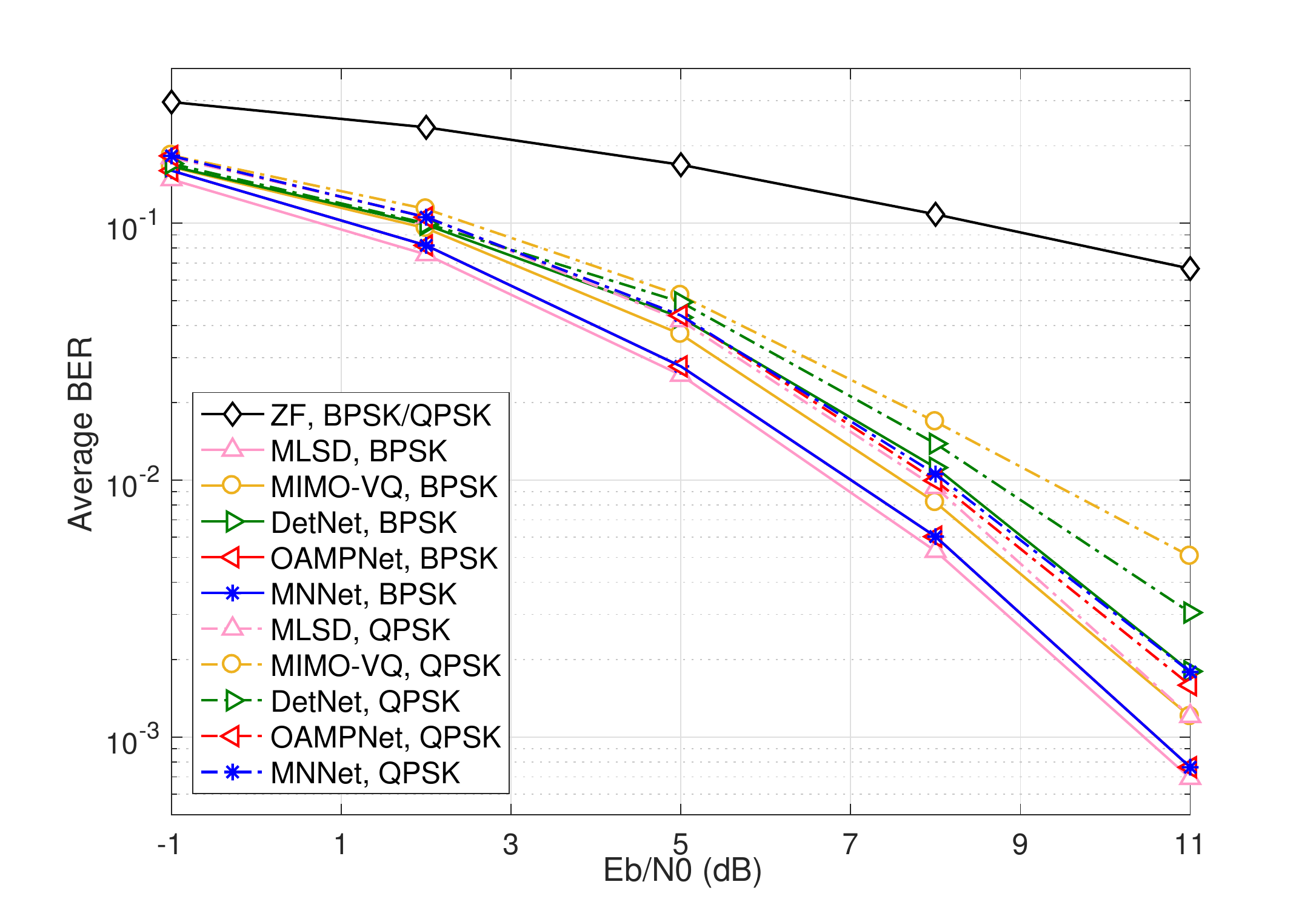}
\caption{\gls{ber} as a function of $E_b/N_0$ for different detection algorithms in uncoded 4-by-4 \gls{mimo} system with BPSK/QPSK modulation.}
\label{fig10}
\end{figure}

\figref{fig10} illustrates the average \gls{ber} performance of different detection algorithms in uncoded 4-by-4 \gls{mimo} system with BPSK/QPSK modulation. Similar results have been observed since all of the deep learning based algorithms largely outperform \gls{zf} receiver across different modulation schemes over a wide range of \gls{snr}s. For BPSK modulation, DetNet with its performance gap between \gls{mlsd} around 1 dB at \gls{ber} of $10^{-2}$; simple \gls{mimo}-\gls{vq} achieves a slightly better performance with a gap around 0.8 dB. OAMPNet and MNNet are both very close to optimum performance at whole \gls{snr} range. For QPSK modulation, DetNet outperforms \gls{mimo}-\gls{vq} with a performance improvement approximately 0.5 dB at \gls{ber} of $10^{-2}$. Meanwhile, OMAPNet and MNNet both achieve near-optimum performance and OMAPNet slightly outperforms MNNet at high \gls{snr} range with a performance gap less than 0.2 dB.

\begin{figure}[t!]
\centering
\includegraphics[width=3.5in]{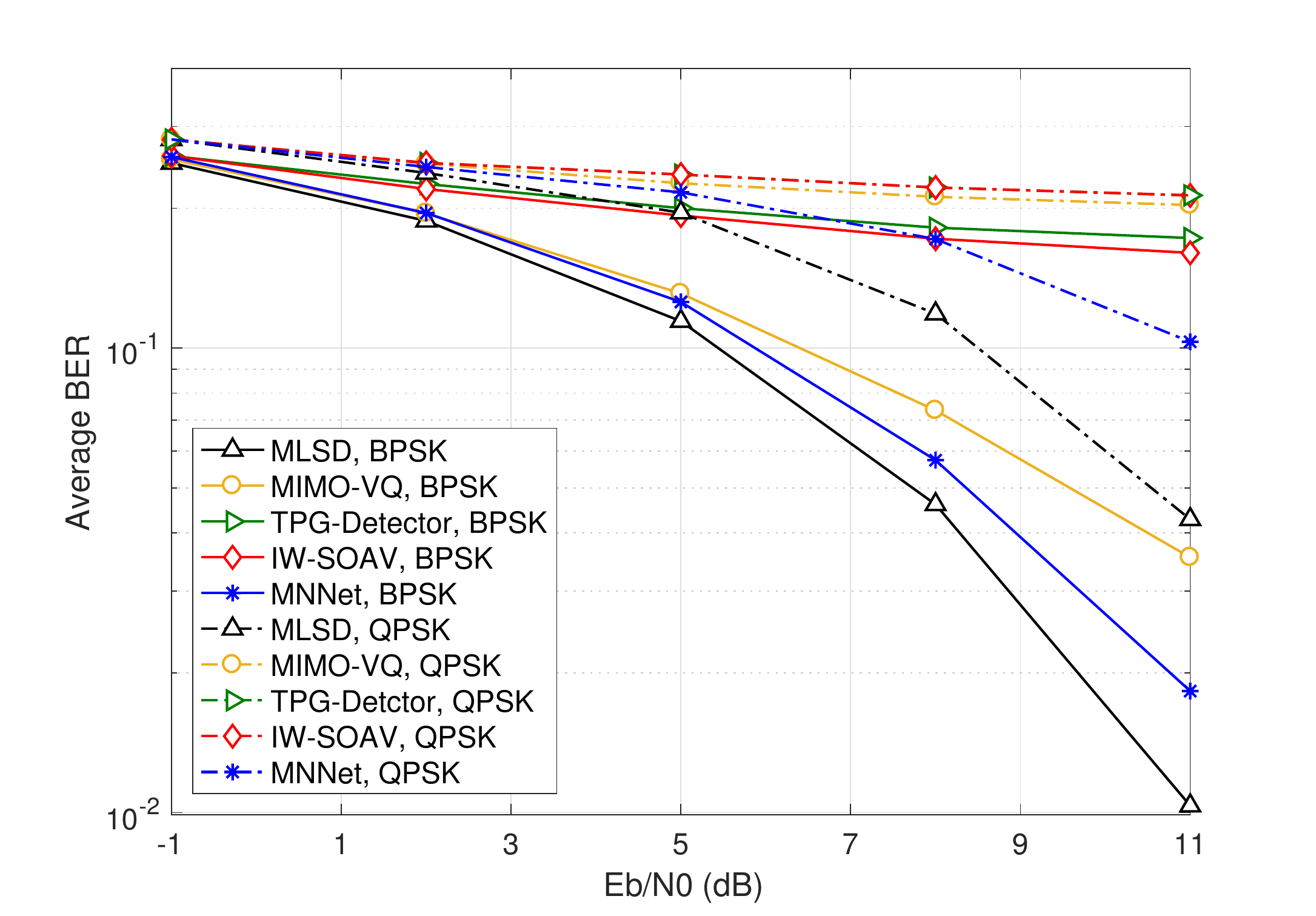}
\caption{\gls{ber} as a function of $E_b/N_0$ for different detection algorithms in uncoded 8-by-4 \gls{mimo} system with BPSK/QPSK modulation.}
\label{fig11}
\end{figure}

\figref{fig11} illustrates the average \gls{ber} performance of different detection algorithms in uncoded 8-by-4 \gls{mimo} system with BPSK/QPSK modulation. The baselines for performance comparison including \gls{mlsd}, IW-SOAV and TPG-Detector. Here we do not employ any channel-inversion based algorithms (e.g. ZF and OAMPNet) due to matrix singularity. For BPSK, it is observed that simple \gls{mimo}-\gls{vq} is able to outperform both IW-SOAV and TPG-Detector; and the proposed MNNet algorithm further improves the performance of \gls{mimo}-\gls{vq} model for approximately 1.7 dB at high \gls{snr} range. Beside, the gap between MNNet and \gls{mlsd} is around 1 dB. However, its performance gap with \gls{mlsd} increases as we move to QPSK (e.g. around 2.3 dB at \gls{ber} of $10^{-1}$). Meanwhile, the other three approaches almost fail to detect with an error floor at high \gls{snr} range due to system outage.}

\section{Conclusion}
\label{sec6}
{\color{sblue} In this paper, we have intensively studied fundamental behaviors of the \gls{ann}-assisted \gls{mimo} detection. It has been revealed that the \gls{ann}-assisted \gls{mimo} detection is naturally a \gls{mimo}-\gls{vq} procedure, which includes joint statistical channel quantization and message quantization. Our mathematical work has shown that the quantization loss of \gls{mimo}-\gls{vq} grows linearly with the number of transmit antennas; and due to this reason, the \gls{ann}-assisted \gls{mimo} detection scales poorly with the size of \gls{mimo}. To tackle the scalability problem, we investigated a \gls{mnn} based  \gls{mimo}-\gls{vq} approach, named MNNet, which can take the advantage of \gls{pic} to scale up the \gls{mimo}-\gls{vq} on each super layer, and consequently the whole network. Computer simulations have demonstrated that the MNNet is able to achieve near-optimum performances in a wide range of communication tasks. In addition, it is shown that a well-modularized network architecture can ensure better computation efficiency in both the \gls{ann} training and evaluating procedures.}

\ifCLASSOPTIONcaptionsoff
  \newpage
\fi

\bibliographystyle{IEEEtran}
\bibliography{./ref}

\begin{thebibliography}{10}
\providecommand{\url}[1]{#1}
\csname url@samestyle\endcsname
\providecommand{\newblock}{\relax}
\providecommand{\bibinfo}[2]{#2}
\providecommand{\BIBentrySTDinterwordspacing}{\spaceskip=0pt\relax}
\providecommand{\BIBentryALTinterwordstretchfactor}{4}
\providecommand{\BIBentryALTinterwordspacing}{\spaceskip=\fontdimen2\font plus
\BIBentryALTinterwordstretchfactor\fontdimen3\font minus
  \fontdimen4\font\relax}
\providecommand{\BIBforeignlanguage}[2]{{%
\expandafter\ifx\csname l@#1\endcsname\relax
\typeout{** WARNING: IEEEtran.bst: No hyphenation pattern has been}%
\typeout{** loaded for the language `#1'. Using the pattern for}%
\typeout{** the default language instead.}%
\else
\language=\csname l@#1\endcsname
\fi
#2}}
\providecommand{\BIBdecl}{\relax}
\BIBdecl

\bibitem{1266912}
A.~J. Paulraj, D.~A. Gore, R.~U. Nabar, and H.~Bolcskei, ``An overview of
  {MIMO} communications - a key to gigabit wireless,'' \emph{Proceedings of the
  IEEE}, vol.~92, no.~2, pp. 198--218, Feb. 2004.

\bibitem{DBLP:journals/corr/OSheaEC17}
T.~J. O'Shea, T.~Erpek, and T.~C. Clancy, ``Deep learning based {MIMO}
  communications,'' \emph{CoRR}, vol. abs/1707.07980, 2017.

\bibitem{8382166}
Q.~{Mao}, F.~{Hu}, and Q.~{Hao}, ``Deep learning for intelligent wireless
  networks: {A} comprehensive survey,'' \emph{IEEE Communications Surveys
  Tutorials}, vol.~20, no.~4, pp. 2595--2621, Fourthquarter 2018.

\bibitem{8755300}
M.~{Chen}, U.~{Challita}, W.~{Saad}, C.~{Yin}, and M.~{Debbah}, ``Artificial
  neural networks-based machine learning for wireless networks: {A} tutorial,''
  \emph{IEEE Communications Surveys Tutorials}, vol.~21, no.~4, pp. 3039--3071,
  Fourthquarter 2019.

\bibitem{8642915}
N.~{Samuel}, T.~{Diskin}, and A.~{Wiesel}, ``Learning to detect,'' \emph{IEEE
  Transactions on Signal Processing}, vol.~67, no.~10, pp. 2554--2564, May
  2019.

\bibitem{DBLP:journals/corr/abs-1812-10044}
S.~Takabe, M.~Imanishi, T.~Wadayama, and K.~Hayashi, ``Trainable projected
  gradient detector for massive overloaded {MIMO} channels: Data-driven tuning
  approach,'' \emph{CoRR}, vol. abs/1812.10044, 2018.

\bibitem{8646357}
H.~{He}, C.~{Wen}, S.~{Jin}, and G.~Y. {Li}, ``A model-driven deep learning
  network for {MIMO} detection,'' in \emph{2018 IEEE GlobalSIP}, Nov. 2018, pp.
  584--588.

\bibitem{DBLP:journals/corr/abs-1901-08329}
Y.~Liu, S.~Bi, Z.~Shi, and L.~Hanzo, ``When machine learning meets big data:
  {A} wireless communication perspective,'' \emph{CoRR}, vol. abs/1901.08329,
  2019.

\bibitem{Takefuji:1992:NNP:128941}
Y.~Takefuji, \emph{Neural Network Parallel Computing}.\hskip 1em plus 0.5em
  minus 0.4em\relax Norwell, MA, USA: Kluwer Academic Publishers, 1992.

\bibitem{DBLP:journals/corr/abs-1812-02858}
J.~Park, S.~Samarakoon, M.~Bennis, and M.~Debbah, ``Wireless network
  intelligence at the edge,'' \emph{CoRR}, vol. abs/1812.02858, 2018.

\bibitem{Yazar2019}
A.~Yazar and H.~Arslan, ``Reliability enhancement in multi-numerology-based
  {5G} new radio using {INI}-aware scheduling,'' \emph{EURASIP Journal on
  Wireless Communications and Networking}, vol. 2019, no.~1, p. 110, May 2019.

\bibitem{4686826}
N.~{Kim}, Y.~{Lee}, and H.~{Park}, ``Performance analysis of {MIMO} system with
  linear {MMSE} receiver,'' \emph{IEEE Trans. Wireless Commun.}, vol.~7,
  no.~11, pp. 4474--4478, Nov. 2008.

\bibitem{5910122}
P.~{Li}, R.~C. {de Lamare}, and R.~{Fa}, ``Multiple feedback successive
  interference cancellation detection for multiuser {MIMO} systems,''
  \emph{IEEE Trans. Wireless Commun.}, vol.~10, no.~8, pp. 2434--2439, Aug.
  2011.

\bibitem{1341066}
S.~M. {Razavizadeh}, V.~T. {Vakili}, and P.~{Azmi}, ``A new faster sphere
  decoder for {MIMO} systems,'' in \emph{Proceedings of the 3rd IEEE ISSPIT},
  Dec. 2003, pp. 86--89.

\bibitem{6375940}
F.~{Rusek}, D.~{Persson}, B.~K. {Lau}, E.~G. {Larsson}, T.~L. {Marzetta},
  O.~{Edfors}, and F.~{Tufvesson}, ``Scaling up {MIMO}: Opportunities and
  challenges with very large arrays,'' \emph{IEEE Signal Process. Mag.},
  vol.~30, no.~1, pp. 40--60, Jan. 2013.

\bibitem{8416771}
L.~{Van der Perre}, L.~{Liu}, and E.~G. {Larsson}, ``Efficient {DSP} and
  circuit architectures for massive {MIMO}: State of the art and future
  directions,'' \emph{IEEE Trans. Signal Process.}, vol.~66, no.~18, pp.
  4717--4736, Sep. 2018.

\bibitem{8054694}
T.~O'Shea and J.~Hoydis, ``An introduction to deep learning for the physical
  layer,'' \emph{IEEE Trans. Cogn. Commun. Netw.}, vol.~3, no.~4, pp. 563--575,
  Dec. 2017.

\bibitem{8761999}
S.~{Xue}, Y.~{Ma}, A.~{Li}, N.~{Yi}, and R.~{Tafazolli}, ``On unsupervised deep
  learning solutions for coherent {MU-SIMO} detection in fading channels,'' in
  \emph{2019 IEEE Int. Conf. Commun.}, May 2019, pp. 1--6.

\bibitem{7244171}
S.~{Yang} and L.~{Hanzo}, ``Fifty years of {MIMO} detection: The road to
  large-scale {MIMO}s,'' \emph{{IEEE} Commun. Surveys Tuts.}, vol.~17, no.~4,
  pp. 1941--1988, Fourthquarter 2015.

\bibitem{10.1007/11802839_32}
D.~S. Poznanovic, ``The emergence of non-von {n}eumann processors,'' in
  \emph{Reconfigurable Computing: Architectures and Applications}.\hskip 1em
  plus 0.5em minus 0.4em\relax Berlin, Heidelberg: Springer Berlin Heidelberg,
  2006, pp. 243--254.

\bibitem{Hawkins2004ThePO}
D.~M. Hawkins, ``The problem of overfitting,'' \emph{Journal of Chemical
  Information and Computer Sciences}, vol.~44, no.~1, pp. 1--12, 2004, pMID:
  14741005.

\bibitem{1181869}
Y.~C. Youngseok~Lee, Kyoungae~Kim, ``Optimization of {AP} placement and channel
  assignment in wireless {LAN}s,'' in \emph{27th Annual IEEE Conf. on LCN},
  Nov. 2002, pp. 831--836.

\bibitem{article1}
C.~Rudin, ``Stop explaining black box machine learning models for high stakes
  decisions and use interpretable models instead,'' \emph{Nat. Mach. Intell.},
  vol.~1, pp. 206--215, May 2019.

\bibitem{article2}
D.~Castelvecchi, ``Can we open the black box of {AI?}'' \emph{Nature}, vol.
  538, pp. 20--23, Oct. 2016.

\bibitem{8466590}
A.~{Adadi} and M.~{Berrada}, ``Peeking inside the black-box: {A} survey on
  explainable artificial intelligence,'' \emph{IEEE Access}, vol.~6, pp.
  52\,138--52\,160, 2018.

\bibitem{Ahalt93}
S.~C. Ahalt and J.~E. Fowler, ``Vector quantization using artificial neural
  network models,'' in \emph{Int. Workshop Adaptive Methods and Emergent Tech.
  for Signal Process. and Commun.}, Jun. 1993, pp. 42--61.

\bibitem{Gersho}
A.~Gersho and R.~M. Gray, \emph{Vector Quantization and Signal
  Compression}.\hskip 1em plus 0.5em minus 0.4em\relax Kluwer international
  series in engineering and computer science, Norwell, MA, 1992.

\bibitem{hebb-organization-of-behavior-1949}
D.~O. Hebb, \emph{The organization of behavior: {A} neuropsychological
  theory}.\hskip 1em plus 0.5em minus 0.4em\relax New York: Wiley, Jun. 1949.

\bibitem{Feng88}
N.~M. Nasrabadi and Y.~Feng, ``Vector quantization of images based on the
  {Kohonen} self-organizing feature maps,'' in \emph{IEEE Int. Conf. Neural
  Netw.}, Jul. 1988, pp. 101--108.

\bibitem{Kohonen90}
T.~Kohonen, ``The self-organizing map,'' \emph{Proc. of the IEEE}, vol.~78, pp.
  1464--1480, Sep. 1990.

\bibitem{Grossberg76}
S.~Grossberg, ``Adaptive pattern classification and universal recoding: {Part
  I}. parallel development and coding of neural feature detectors,''
  \emph{Biological Cybernetics}, vol.~23, pp. 121--134, 1976.

\bibitem{Carlos18}
{J. C. De Luna Ducoing}, Y.~Ma, N.~Yi, and R.~Tafazolli, ``A real complex
  hybrid modulation approach for scaling up multiuser {MIMO} detection,''
  \emph{IEEE Trans. Commun.}, vol.~56, pp. 3916--3929, Sep. 2018.

\bibitem{6313426}
J.~M. {Keller}, M.~R. {Gray}, and J.~A. {Givens}, ``{A} fuzzy k-nearest
  neighbor algorithm,'' \emph{IEEE Trans. Syst., Man, Cybern.}, vol. SMC-15,
  no.~4, pp. 580--585, Jul. 1985.

\bibitem{10.1117/12.279444}
J.~P. Theiler and G.~Gisler, ``{Contiguity-enhanced k-means clustering
  algorithm for unsupervised multispectral image segmentation},'' in
  \emph{Algorithms, Devices, and Systems for Optical Information Processing},
  vol. 3159, International Society for Optics and Photonics.\hskip 1em plus
  0.5em minus 0.4em\relax SPIE, 1997, pp. 108 -- 118.

\bibitem{3776}
N.~M. {Nasrabadi} and R.~A. {King}, ``Image coding using vector quantization:
  {a} review,'' \emph{IEEE Trans. Commun.}, vol.~36, no.~8, pp. 957--971, Aug.
  1988.

\bibitem{Bovik:2005:HIV:1200338}
A.~C. Bovik, \emph{Handbook of Image and Video Processing}.\hskip 1em plus
  0.5em minus 0.4em\relax Orlando, FL, USA: Academic Press, Inc., 2005.

\bibitem{Goodfellow-et-al-2016}
I.~Goodfellow, Y.~Bengio, and A.~Courville, \emph{Deep Learning}.\hskip 1em
  plus 0.5em minus 0.4em\relax MIT Press, 2016.

\bibitem{Happel94thedesign}
B.~L. Happel and J.~M.~J. Murre, ``The design and evolution of modular neural
  network architectures,'' \emph{Neural Networks}, vol.~7, pp. 985--1004, 1994.

\bibitem{563486}
M.~{Latva-aho} and J.~{Lilleberg}, ``Parallel interference cancellation in
  multiuser detection,'' in \emph{Proceedings of International Symposium on
  Spread Spectrum Techniques and Applications}, vol.~3, Sep. 1996, pp.
  1151--1155.

\bibitem{888}
J.~Schmidhuber, ``Deep learning in neural networks: An overview,'' \emph{Neural
  Networks}, vol.~61, pp. 85--117, 2015.

\bibitem{Rumelhart:1995:BBT:201784.201785}
D.~E. Rumelhart, R.~Durbin, R.~Golden, and Y.~Chauvin,
  ``Backpropagation.''\hskip 1em plus 0.5em minus 0.4em\relax Hillsdale, NJ,
  USA: L. Erlbaum Associates Inc., 1995, ch. Backpropagation: The Basic Theory,
  pp. 1--34.

\bibitem{Riedmiller1994AdvancedSL}
M.~A. Riedmiller, ``Advanced supervised learning in multi-layer perceptrons -
  from backpropagation to adaptive learning algorithms,'' \emph{Computer
  Standards and Interfaces}, vol.~16, pp. 265--278, 1994.

\bibitem{luo2018adaptive}
L.~Luo, Y.~Xiong, and Y.~Liu, ``Adaptive gradient methods with dynamic bound of
  learning rate,'' in \emph{International Conference on Learning
  Representations}, 2019.

\bibitem{5534810}
S.~G. {Kim}, D.~{Yoon}, S.~K. {Park}, and Z.~{Xu}, ``Performance analysis of
  the {MIMO} zero-forcing receiver over continuous flat fading channels,'' in
  \emph{The 3rd International Conference on Information Sciences and
  Interaction Sciences}, June 2010, pp. 324--327.

\bibitem{DBLP:journals/corr/abs-1102-1462}
A.~H. Mehana and A.~Nosratinia, ``Diversity of {MMSE} {MIMO} receivers,''
  \emph{CoRR}, vol. abs/1102.1462, 2011.

\bibitem{8254935}
Y.~{Hama} and H.~{Ochiai}, ``Performance analysis of matched filter detector
  for {MIMO} systems in rayleigh fading channels,'' in \emph{GLOBECOM 2017 -
  2017 IEEE Global Communications Conference}, Dec 2017, pp. 1--6.

\bibitem{7760475}
R.~{Hayakawa}, K.~{Hayashi}, H.~{Sasahara}, and M.~{Nagahara}, ``Massive
  overloaded {MIMO} signal detection via convex optimization with proximal
  splitting,'' in \emph{2016 24th European Signal Processing Conference
  (EUSIPCO)}, Aug 2016, pp. 1383--1387.

\end{thebibliography}
\end{document}